%% file: mr.tex
\documentclass{imsart}

\usepackage{amsthm,amsmath,amssymb}
\RequirePackage[numbers]{natbib}
\RequirePackage[colorlinks,citecolor=blue,urlcolor=blue]{hyperref}

\usepackage{enumerate}
\usepackage{float}
\usepackage{color}
\usepackage{booktabs,rotating}
\usepackage[nolists]{endfloat}
\usepackage{endrotfloat}
\usepackage{bm}

\startlocaldefs

\def \E{\mathbb{E}}
\def \P{\mathbb{P}}
\def \G{\mathbb{G}}
\def \R{\mathbb{R}}

\def\vec{\bm}
\newcommand{\1}{\mathbf{1}}
\newcommand{\as}{\xrightarrow{\text{a.s.}}}
\newcommand{\p}{\rightarrow_p}
\newcommand{\dd}{\mathrm{d}}
\newcommand{\CC}{\mathcal{C}}
\newcommand{\DD}{\mathcal{D}}
\newcommand{\FF}{\mathcal{F}}
\newcommand{\GG}{\mathcal{G}}
\newcommand{\XX}{\mathcal{X}}
\newcommand{\NN}{\mathcal{N}}
\newcommand{\var}{\mathrm{V}}
\newcommand{\cov}{\mathrm{cov}}
\renewcommand{\Pr}{\mathrm{Pr}}

\newtheorem{prop}{Proposition}[section]

\newtheorem{lem}{Lemma}[section]

\newcount\Comments  
\newcommand{\kibitz}[2]{\ifnum\Comments=1\textcolor{#1}{#2}\fi}

\endlocaldefs

\begin{document}
\begin{frontmatter}

\title{Semiparametric estimation of a two-component mixture of linear regressions in which one component is known}
\runtitle{Mixture of linear regressions in which one component is known}

\begin{aug}
  \author{L. Bordes
    \ead[label=e1]{laurent.bordes@univ-pau.fr}},
  \author{I. Kojadinovic\ead[label=e2]{ivan.kojadinovic@univ-pau.fr}}
  
  \address{Universit\'e de Pau et des Pays de l'Adour, Laboratoire de math\'ematiques et de leurs applications, UMR CNRS 5142, B.P. 1155, 64013 Pau Cedex, France\\ \printead{e1,e2}}
  
  \author{P. Vandekerkhove
    \ead[label=e3]{pierre.vandekerkhove@univ-mlv.fr}}
  
  \address{Universit\'e Paris-Est, Marne-la-Vall\'ee, Laboratoire d'analyse et de math\'ematiques appliqu\'ees, UMR CNRS 8050, France, and Georgia Institute of Technology, UMI Georgia Tech CNRS 2958, G.W. Woodruff School of Mechanical Engineering, Atlanta, USA\\
    \printead{e3}}


  \runauthor{Bordes et al.}

\end{aug}

\begin{abstract}
A new estimation method for the two-component mixture model introduced in~\cite{Van13} is proposed. This model consists of a two-component mixture of linear regressions in which one component is entirely known while the proportion, the slope, the intercept and the error distribution of the other component are unknown. In spite of good performance for datasets of reasonable size, the method proposed in~\cite{Van13} suffers from a serious drawback when the sample size becomes large as it is based on the optimization of a contrast function whose pointwise computation requires $O(n^2)$ operations. The range of applicability of the method derived in this work is substantially larger as it relies on a method-of-moments estimator free of tuning parameters whose computation requires $O(n)$ operations. From a theoretical perspective, the asymptotic normality of both the estimator of the Euclidean parameter vector and of the semiparametric estimator of the c.d.f.\ of the error is proved under weak conditions not involving zero-symmetry assumptions. In addition, an approximate confidence band for the c.d.f.\ of the error can be computed using a weighted bootstrap whose asymptotic validity is proved. The finite-sample performance of the resulting estimation procedure is studied under various scenarios through Monte Carlo experiments. The proposed method is illustrated on three real datasets of size $n=150$, 51 and 176,343, respectively. Two extensions of the considered model are discussed in the final section: a model with an additional scale parameter for the first component, and a model with more than one explanatory variable. 
\end{abstract} 




\tableofcontents

\end{frontmatter}

\section{Introduction}

Practitioners are frequently interested in modeling the relationship between a random response variable $Y$ and a $d$-dimensional random explanatory vector $\vec X$ by means of a linear regression model estimated from a random sample $(\vec X_i,Y_i)_{1 \leq i \leq n}$ of $(\vec X,Y)$. Quite often, the homogeneity assumption claiming that the linear regression coefficients are the same for all the observations $(\vec X_1,Y_1), \dots,(\vec X_n,Y_n)$ is inadequate. To allow different parameters for different groups of observations, a Finite Mixture of Regressions (FMR) can be considered; see~\cite{GruLei06,Lei04} for nice overviews. 

Statistical inference for the fully parametric FMR model was first considered in~\cite{QuaRam78} where an estimation method based on the moment generating function was proposed. An EM estimating approach was investigated in~\cite{DeV89} in the case of two components. Variations of the latter approach were also considered in~\cite{JonMcL92} and~\cite{Tur00}. The problem of determining the number of components in the parametric FMR model was investigated in~\cite{HawAllSto01} using methods derived from the likelihood equation. In~\cite{HurJusRob03}, the authors proposed a Bayesian approach to estimate the regression coefficients and also considered an extension of the model in which the number of components is unspecified. The asymptotics of maximum likelihood estimators of parametric FMR models were studied in~\cite{ZhuZha04}. More recently, an $\ell_1$-penalized method based on a Lasso-type estimator for a high-dimensional FMR model with $d\gg n$ was proposed in~\cite{StaBuhvan10}. 

As an alternative to parametric estimation of a FMR model, some authors suggested the use of more flexible semiparametric approaches. This research direction finds its origin in~\cite{HalZho03} where $d$-variate semiparametric mixture models of random vectors with independent components were considered. The authors showed in particular that, for $d\geq 3$, it is possible to identify a two-component model without parametrizing the distributions of the component random vectors. To the best of our knowledge, Leung and Qin~\cite{LeuQin06} were the first to estimate a FMR model semiparametrically. In the two-component case, they studied the situation in which the components are related by Anderson's exponential tilt model~\cite{And79}. Hunter and Young~\cite{HunYou12} studied the identifiability of an $m$-component semiparametric FMR model and numerically investigated an EM-type algorithm for estimating its parameters. Vandekerkhove~\cite{Van13} proposed an $M$-estimation method for a two-component semiparametric mixture of regressions with symmetric errors in which one component is known. The latter approach was applied to data extracted from a high-density microarray and modeled in~\cite{MarMarBer08} by means of a parametric FMR. 

The semiparametric approach proposed in~\cite{Van13} is of interest for two main reasons. Due to its semiparametric nature, the method allows to detect complex structures in the error of the unknown regression component. It can additionally be regarded as a tool to assess the relevance of results obtained using EM-type algorithms. The approach has however three important drawbacks. First, it is not theoretically valid when the errors are not symmetric. Second, it is very computationally expensive for large datasets as it requires the optimization of a contrast function whose pointwise evaluation requires $O(n^2)$ operations. Third, the underlying optimization method requires the choice of a weight function and initial values for the Euclidean parameters, neither choices being data-driven.  

The object of interest of this paper is the two-component FMR model studied in~\cite{Van13} in which one component is entirely known while the proportion, the slope, the intercept and the error distribution of the other component are unknown. The estimation of the Euclidean parameter vector is achieved through the method of moments. Semiparametric estimators of the c.d.f.\ and the p.d.f.\ of the error of the unknown component are proposed. The proof of the asymptotic normality of the Euclidean and functional estimators is not based on zero-symmetry-like assumptions frequently found in the literature but only involves finite moments of order eight for the explanatory variable and the boundness of the p.d.f.s of the errors and their derivatives. The almost sure uniform consistency of the estimator of the p.d.f.\ of the unknown error is obtained under similar conditions. A consequence of these theoretical results is that, unlike for EM-type approaches, the estimation uncertainty can be assessed through large-sample standard errors for the Euclidean parameters and by means of an approximate confidence band for the c.d.f.\ of the unknown error. The latter is computed using a weighted bootstrap whose asymptotic validity is proved. 

From a practical perspective, it is worth mentioning that the range of applicability of the resulting semiparametric estimation procedure is substantially larger than the one proposed in~\cite{Van13} as its computation only requires $O(n)$ operations  and no tuning of parameters such as starting values or weight functions. As a consequence, very large datasets can be easily processed. For instance, as shall be seen in Section~\ref{illus}, the estimation of the parameters of the model from the ChIPmix data considered in~\cite{MarMarBer08} consisting of $n=176,343$ observations took less than 30 seconds on one 2.4 GHz processor. The estimation of the same model from a subset of $n=30,000$ observations using the method in~\cite{Van13} took more than two days on a similar processor.

The paper is organized as follows. Section~\ref{PN} is devoted to a detailed description of the model, while Section~\ref{ident} is concerned with its identifiability. The estimators of the Euclidean parameter vector and of the functional parameter are investigated in detail in Section~\ref{estimation}. The finite-sample performance of the proposed estimation method is studied for various scenarios through Monte Carlo experiments in Section~\ref{mc}. In Section~\ref{illus}, the proposed method is applied to the tone data analyzed, among others, in~\cite{HunYou12}, to the aphids dataset studied initially in~\cite{BoiSinSinTaiTur98}, and to the ChIPmix data considered in~\cite{MarMarBer08}. Two extensions of the  FMR model under consideration are discussed in the last section: a model with an additional scale parameter for the first component, and a model with more than one explanatory variable.

Note finally that all the computations reported in this work were carried out using the R statistical system~\cite{Rsystem} and that the main corresponding R functions are available on the web page of the second author. 

\section{Problem and notation}
\label{PN}

Let $Z$ be a Bernoulli random variable with unknown parameter $\pi_0 \in [0,1]$, let $X$ be an $\XX$-valued random variable with $\XX \subset \R$, and let $\varepsilon^*,\varepsilon^{**}$ be two absolutely continuous centered real valued random variables with finite variances and independent of $X$. Assume additionally that $Z$ is independent of $X$, $\varepsilon^*$ and $\varepsilon^{**}$. Furthermore, for fixed $\alpha_0^*,\beta_0^*,\alpha_0^{**},\beta_0^{**} \in \R$, let $\tilde Y$ be the random variable defined by
$$
\tilde Y = (1-Z)(\alpha_0^*+\beta_0^* X+\varepsilon^*)+Z(\alpha_0^{**}+\beta_0^{**} X+\varepsilon^{**}),
$$
i.e.,
\begin{equation}
\label{twocomp}
\tilde Y = \left\{
\begin{array}{lll}
\alpha_0^*+\beta_0^*X+\varepsilon^* &\mbox{ if } & Z=0, \\
\alpha_0^{**}+\beta_0^{**}X+\varepsilon^{**} &\mbox{ if } & Z=1.
\end{array}\right. 
\end{equation}
The above display is the equation of a mixture of two linear regressions with $Z$ as mixing variable.

Let $F^*$ and $F^{**}$ denote the c.d.f.s of $\varepsilon^*$ and $\varepsilon^{**}$, respectively. Furthermore, $\alpha_0^*$, $\beta_0^*$ and $F^*$ are assumed known while $\alpha_0^{**}$, $\beta_0^{**}$, $\pi_0$ and $F^{**}$ are assumed unknown. The aim of this work is to propose and study an estimator of $(\alpha_0^{**},\beta_0^{**},\pi_0,F^{**})$ based on $n$ i.i.d.\ copies $(X_i,\tilde Y_i)_{1\leq i\leq n}$ of $(X,\tilde Y)$. Now, define $Y=\tilde Y-\alpha_0^*-\beta_0^* X$, $\alpha_0=\alpha_0^{**}-\alpha_0^*$ and $\beta_0=\beta_0^{**}-\beta_0^*$, and notice that 
\begin{equation}
\label{model}
 Y = \left\{
\begin{array}{lll}
\varepsilon^* &\mbox{ if } & Z=0, \\
\alpha_0+\beta_0 X+\varepsilon &\mbox{ if } & Z=1,
\end{array}\right. 
\end{equation}
where, to simplify the notation, $\varepsilon = \varepsilon^{**}$ and $F= F^{**}$. It follows that the previous estimation problem is equivalent to the problem of estimating $(\alpha_0,\beta_0,\pi_0,F)$ from the observation of $n$ i.i.d.\ copies $(X_i,Y_i)_{1\leq i\leq n}$ of $(X,Y)$.\\

As we continue, the unknown c.d.f.s of $X$ and $Y$ will be denoted by $F_X$ and $F_Y$, respectively. Also, for any $x \in \XX$, the conditional c.d.f.\ of $Y$ given $X=x$ will be denoted by $F_{Y|X}(\cdot|x)$, and we have
\begin{equation}
\label{FYX}
F_{Y|X}(y|x)=(1-\pi_0)F^*(y)+\pi_0 F(y-\alpha_0-\beta_0 x), \qquad y\in\R. 
\end{equation}
It follows that, for any $x \in \XX$, $f_{Y|X}(\cdot|x)$, the conditional p.d.f.\ of $Y$ given $X=x$, can be expressed as
\begin{equation}
\label{fYX}
f_{Y|X}(y|x)=(1-\pi_0)f^*(y)+\pi_0 f(y-\alpha_0-\beta_0 x), \qquad y\in\R,
\end{equation}
where $f^*$ and $f$ are the p.d.f.s of $\varepsilon^*$ and $\varepsilon$, assuming that they exist on $\R$.

Note that, as shall be discussed in Section~\ref{extension}, it is possible to consider a slightly more general version of the model stated in~\eqref{model} involving an unknown scale parameter for the first component. This more elaborate model remains identifiable and estimation through the method of moments is theoretically possible. However, from a practical perspective, estimation of this scale parameter through the method of moments seems quite unstable insomuch as that an alternative estimation method appears to be required. Notice also that another more straightforward extension of the model will be considered in Section~\ref{extension} allowing to deal with more than one explanatory variable.

\section{Identifiability}
\label{ident}

Since~(\ref{model}) is clearly equivalent to
\begin{equation}
\label{Y}
Y = (1-Z) \varepsilon^* + Z (\alpha_0 + \beta_0 X + \varepsilon),
\end{equation}
we immediately obtain that
\begin{equation}
\label{M1}
\E(Y|X) = \pi_0\alpha_0 + \pi_0\beta_0 X \qquad \mbox{a.s.}
\end{equation}
It follows that the coefficients $\lambda_{0,1}=\pi_0\alpha_0$ and $\lambda_{0,2}=\pi_0\beta_0$ can be identified from~(\ref{M1}) if $|\XX| > 1$. In addition, we have
\begin{eqnarray}
\E(Y^2|X)&=&\E[\{(1-Z)\varepsilon^*+Z(\alpha_0+\beta_0 X+\varepsilon)\}^2|X]\quad \mbox{a.s.}\nonumber \\
&=& \E(1-Z)\E\{(\varepsilon^*)^2\}+\E(Z)\E\{(\alpha_0+\beta_0 X)^2+\varepsilon^2|X\}\quad \mbox{a.s.}\nonumber \\
&=& (1-\pi_0)(\sigma_0^*)^2+\pi_0\left(\alpha_0^2+2\alpha_0\beta_0 X+\beta_0^2X^2+\sigma_0^2\right)\quad \mbox{a.s.}\nonumber \\
&=& (1-\pi_0)(\sigma_0^*)^2+\pi_0(\alpha_0^2+\sigma_0^2)+ 2\pi_0\alpha_0\beta_0 X +\pi_0\beta_0^2X^2 \, \mbox{a.s.}, \label{M2}
\end{eqnarray}
where $\sigma_0^*$ and $\sigma_0$ are the standard deviations of $\varepsilon^*$ and $\varepsilon$, respectively. If $\XX$ contains three points ${x_1,x_2,x_3}$ such that the vectors $\{(1,x_1,x^2_1), (1,x_2,x^2_2), (1,x_3,x^2_3)\}$ are linearly independent then, from (\ref{M2}), we can identify the coefficients  $\lambda_{0,3}=(1-\pi_0)(\sigma_0^*)^2+\pi_0(\alpha_0^2+\sigma_0^2)$, $\lambda_{0,4}=2\pi_0\alpha_0\beta_0$ and $\lambda_{0,5}=\pi_0\beta_0^2$. It then remains to identify $\alpha_0$, $\beta_0$ and $\pi_0$ from the equations
\begin{equation}
\label{relations}
\left\{
\begin{array}{lll}
\lambda_{0,1} & = & \pi_0\alpha_0  \\
\lambda_{0,2} & = & \pi_0\beta_0 \\
\lambda_{0,3} & = & (1-\pi_0)(\sigma_0^*)^2+\pi_0(\alpha_0^2+\sigma_0^2) \\
\lambda_{0,4} & = & 2\pi_0\alpha_0\beta_0 = 2\alpha_0\lambda_{0,2}\\
\lambda_{0,5} & = & \pi_0\beta_0^2 = \beta_0\lambda_{0,2}. 
\end{array}\right.
\end{equation}
From the above system, we see that $\alpha_0$, $\beta_0$ and $\pi_0$ can be identified provided $\pi_0\beta_0 \neq 0$. If $\pi_0 = 0$, then $\alpha_0$ and $\beta_0$ cannot be identified, and, as shall become clear in the sequel, neither can $F$. If $\beta_0 = 0$, then the model in~\eqref{model} coincides with the model studied in~\cite{BorDelVan06} where it was shown that identifiability does not necessary hold even if $\varepsilon^*$ is assumed to have a zero-symmetric distribution. It follows that for identifiability to hold it is necessary that the unknown component actually exists ($\pi_0 \in (0,1]$) and that its slope is non-zero ($\beta_0 \neq 0$). The latter conditions will be assumed in the rest of the paper.

Before discussing the identifiability of the functional part of the model, it is important to notice that the conditions on $\XX$ stated above are merely sufficient conditions. For instance, if $\XX = \{-1,1\}$, then $\lambda_{0,1}=\pi_0\alpha_0$ and $\lambda_{0,2}=\pi_0\beta_0$ can be identified from~(\ref{M1}) and $\lambda_{0,4} = 2\pi_0\alpha_0\beta_0$ can be identified from~(\ref{M2}), which is enough to uniquely determine $(\alpha_0,\beta_0,\pi_0)$.

Let us finally consider the functional part $F$ of the model. For any $\vec \eta=(\alpha,\beta) \in \R^2$, denote by $J(\cdot,\vec \eta)$ the c.d.f.\ defined by 
\begin{equation}
\label{J}
J(t,\vec \eta)=\Pr(Y\leq t+\alpha+\beta X), \qquad t\in\R. 
\end{equation}
For any $t \in \R$, this can be rewritten as
\begin{align*}
J(t,\vec \eta)=& \int_{\R}F_{Y|X}(t+\alpha+\beta x|x) \dd F_X(x)\\
=&(1-\pi_0)\int_{\R}F^*(t+\alpha+\beta x) \dd F_X(x)\\ &+\pi_0\int_{\R}F \{ t+(\alpha-\alpha_0)+(\beta-\beta_0) x \} \dd F_X(x).
\end{align*} 
For $\vec \eta = \vec \eta_0 = (\alpha_0,\beta_0)$, we then obtain
$$
J(t,\vec \eta_0) = (1-\pi_0)\int_{\R}F^*(t+\alpha_0+\beta_0 x) \dd F_X(x)+\pi_0 F(t), \qquad t \in \R.
$$
Now, for any $\vec \eta \in \R^2$, let $K(\cdot,\vec \eta)$ be defined by
\begin{equation}
\label{K}
K(t,\vec \eta)=\int_{\R}F^*(t+\alpha+\beta x) \dd F_X(x), \qquad t \in \R.
\end{equation}
It follows that $F$ is identified since
\begin{equation}
\label{F}
F(t)=\frac{1}{\pi_0} \left\{ J(t,\vec \eta_0)-(1-\pi_0) K(t,\vec \eta_0)  \right\}, \qquad t \in \R.
\end{equation}
The above equation is at the root of the derivation of an estimator for $F$.

\section{Estimation}
\label{estimation}

Let $P$ be the probability distribution of $(X,Y)$. For ease of exposition, we will frequently use the notation adopted in the theory of empirical processes as presented in~\cite{Kos08,van98,vanWel96} for instance. Given a measurable function $f:\R^2 \to \R^k$, for some integer $k \geq 1$, $P f$ will denote the integral $\int f \dd P$. Also, the empirical measure obtained from the random sample $(X_i,Y_i)_{1 \leq i \leq n}$ will be denoted by $\P_n=n^{-1}\sum_{i=1}^n\delta_{X_i,Y_i}$, where $\delta_{x,y}$ is the probability distribution that assigns a mass of 1 at $(x,y)$. The expectation of $f$ under the empirical measure is then $\P_n f = n^{-1} \sum_{i=1}^n\ f(X_i,Y_i)$ and the quantity $\G_n f = \sqrt{n} (\P_n f - P f)$ is the {\em empirical process} evaluated at $f$. The arrow~`$\leadsto$' will be used to denote weak convergence in the sense of Definition~1.3.3 in~\cite{vanWel96} and, for any set $S$, $\ell^\infty(S)$ will stand for the space of all bounded real-valued functions on $S$ equipped with the uniform metric. Key results and more details can be found for instance in~\cite{Kos08, van98, vanWel96}.
 
\subsection{Estimation of the Euclidean parameter vector}
\label{euclidean-parameter}

To estimate the Euclidean parameter vector $(\alpha_0,\beta_0,\pi_0) \in \R \times \R \setminus \{0\} \times (0,1]$, we first need to estimate the vector $\vec\lambda_0 = (\lambda_{0,1},\dots,\lambda_{0,5}) \in \R^5$ whose components were expressed in terms of $\alpha_0$, $\beta_0$ and $\pi_0$ in \eqref{relations}. From~(\ref{M1}) and~(\ref{M2}), it is natural to consider the regression function 
$$
d_n(\vec\lambda)=\P_n\varphi_{\vec\lambda}, \qquad \vec\lambda \in \R^5,
$$
where, for any $ \vec\lambda \in \R^5$, $\varphi_{\vec\lambda}:\R^2 \to \R$ is defined by
\begin{equation}
\label{varphi_lambda}
\varphi_{\vec\lambda}(x,y)=(y-\lambda_1-\lambda_2x)^2+(y^2-\lambda_3-\lambda_4x-\lambda_5x^2)^2, \qquad x,y \in \R.
\end{equation}
As an estimator of $\vec\lambda_0=\arg\min_{\vec \lambda} P \varphi_{\vec\lambda}$, we then naturally consider $\vec\lambda_n=\arg\min_{\vec \lambda}d_n(\vec\lambda)$ that satisfies 
$$
\dot{d}_n(\vec\lambda_n)=\P_n\dot \varphi_{\vec\lambda_n}=0,
$$
where $\dot \varphi_{\vec\lambda}$,  the gradient  of $\varphi_{\vec\lambda}$ with respect to $\vec \lambda$, is given by
$$
\dot \varphi_{\vec\lambda}(x,y)=-2\left(
\begin{array}{c}
y-\lambda_1-\lambda_2 x\\
x(y-\lambda_1-\lambda_2 x)\\
y^2-\lambda_3-\lambda_4 x-\lambda_5x^2\\
x(y^2-\lambda_3-\lambda_4 x-\lambda_5x^2)\\
x^2(y^2-\lambda_3-\lambda_4 x-\lambda_5x^2)\\
\end{array}\right), \qquad x,y \in \R.
$$
Now, for any integers $p,q \geq 1$, define
$$
\overline{X^pY^q}=\frac{1}{n}\sum_{i=1}^nX_i^pY_i^q,
$$
and let
$$
\Lambda_n=2\left(
\begin{array}{ccccc}
1  & \overline{X}   & 0  & 0 & 0 \\
\overline{X}  & \overline{X^2} & 0  & 0 & 0 \\
0 & 0 & 1 & \overline{X} & \overline{X^2}  \\
0 & 0 & \overline{X} & \overline{X^2} & \overline{X^3}\\
0 & 0 & \overline{X^2} & \overline{X^3} & \overline{X^4}\\
\end{array}\right)
\quad\mbox{ and }\quad
\vec\Upsilon_n=2\left(
\begin{array}{c}
\overline{Y}\\
\overline{XY}\\
\overline{Y^2}\\
\overline{XY^2}\\
\overline{X^2Y^2}
\end{array}\right),
$$
which respectively estimate
{\small $$
\Lambda_0=2\left(
\begin{array}{ccccc}
1  & \E(X)   & 0  & 0 & 0 \\
 \E(X) & \E(X^2) & 0  & 0 & 0 \\
0 & 0 & 1 &  \E(X) &  \E(X^2)  \\
0 & 0 &  \E(X) & \E(X^2) & \E(X^3)\\
0 & 0 & \E(X^2) & \E(X^3) & \E(X^4)\\
\end{array}\right)
\mbox{ and }
\vec\Upsilon_0=2\left(
\begin{array}{c}
\E(Y)\\
\E(XY)\\
\E(Y^2)\\
\E(XY^2)\\
\E(X^2Y^2)
\end{array}\right).
$$}
The linear equation $\P_n\dot \varphi_{\vec\lambda_n}=0$ can then equivalently be rewritten as $\Lambda_n \vec \lambda_n = \vec\Upsilon_n$. Provided the matrices $\Lambda_n$ and $\Lambda_0$ are invertible, we can write $\vec\lambda_n=\Lambda_n^{-1}\vec\Upsilon_n$ and $\vec\lambda_0=\Lambda_0^{-1}\vec\Upsilon_0$. Notice that, in practice, this amounts to performing an ordinary least-squares linear regression of $Y$ on $X$ to obtain $\lambda_{n,1}$ and $\lambda_{n,2}$, while $\lambda_{n,3}$, $\lambda_{n,4}$ and $\lambda_{n,5}$ are given by an ordinary least-squares linear regression of $Y^2$ on $X$ and $X^2$.

To obtain an estimator of $(\alpha_0,\beta_0,\pi_0)$, we use the relationships induced by~(\ref{M1}) and~(\ref{M2}) and recalled in~(\ref{relations}). Leaving the third equation aside because it involves the unknown standard deviation $\sigma_0$ of~$\varepsilon$, we obtain three possible estimators of~$\alpha_0$:
$$
\alpha_n^{(1)} = \frac{\lambda_{n,1} \lambda_{n,5}}{\lambda_{n,2}^2},  \qquad \alpha_n^{(2)} = \frac{\lambda_{n,4}}{2 \lambda_{n,2}}, \qquad \mbox{or} \qquad \alpha_n^{(3)} = \frac{\lambda_{n,4}^2}{4 \lambda_{n,1} \lambda_{n,5}},
$$
three possibles estimators of $\beta_0$:
$$
\beta_n^{(1)} = \frac{\lambda_{n,5}}{\lambda_{n,2}},  \qquad \beta_n^{(2)} = \frac{\lambda_{n,4}}{2 \lambda_{n,1}}, \qquad \mbox{or} \qquad \beta_n^{(3)} = \frac{\lambda_{n,2} \lambda_{n,4}^2}{4 \lambda_{n,5} \lambda_{n,1}^2},
$$
and, three possibles estimators of $\pi_0$:
$$
\pi_n^{(1)} = \frac{\lambda_{n,2}^2}{\lambda_{n,5}},  \qquad \pi_n^{(2)} = \frac{2 \lambda_{n,1} \lambda_{n,2}}{\lambda_{n,4}}, \qquad \mbox{or} \qquad \pi_n^{(3)} = \frac{4 \lambda_{n,1}^2 \lambda_{n,5}}{\lambda_{n,4}^2}.
$$
There are therefore 27 possible estimators of $(\alpha_0,\beta_0,\pi_0)$. Their asymptotics can be obtained under reasonable conditions similar to those stated in Assumptions A1 and A2 below. Unfortunately, all 27 estimators turned out to behave quite poorly in small samples. This prompted us to look for alternative estimators within the ``same class''. 

We now describe an estimator of $(\alpha_0,\beta_0,\pi_0)$ that was obtained empirically and that behaves significantly better for small samples than the aforementioned ones. The new regression function under consideration is $d_n(\vec\gamma)=\P_n\varphi_{\vec\gamma}$, $\vec\gamma\in\R^8$, where, for any $(x,y)\in\R^2$,
\begin{multline}
\label{varphi_gamma}
\varphi_{\vec \gamma}(x,y)=(y-\gamma_1-\gamma_2x)^2+(y^2-\gamma_3-\gamma_4x^2)^2 \\ +(x-\gamma_5)^2+(x^2-\gamma_6)^2+(x^3-\gamma_7)^2+(x^4-\gamma_8)^2.
\end{multline}
As previously, let $\vec\gamma_n=\arg\min_{\gamma} d_n(\vec\gamma)$ be the estimator of $\vec\gamma_0=\arg\min_{\gamma} P \varphi_{\vec \gamma}$, and notice that the main difference between the approach based on~\eqref{varphi_lambda} and the approach based on~\eqref{varphi_gamma} is that the former involves the linear regression of $Y$ on $X$ and $X^2$, while the latter relies on the linear regression of $Y^2$ on $X^2$ only, which appears to result in better estimation accuracy. Now, let
$$
\Gamma_n=2\left(\begin{array}{cccccccc}
1 & \overline{X} & 0 & 0 & 0 & 0 & 0 & 0 \\
\overline{X} & \overline{X^2} & 0 & 0 & 0 & 0 & 0 & 0 \\
0 & 0 & 1 & \overline{X^2} & 0 & 0 & 0 & 0 \\
0 & 0 & \overline{X^2} & \overline{X^4} & 0 & 0 & 0 & 0 \\
0 & 0 & 0 & 0 & 1 & 0 & 0 & 0 \\
0 & 0 & 0 & 0 & 0 & 1 & 0 & 0 \\
0 & 0 & 0 & 0 & 0 & 0 & 1 & 0 \\
0 & 0 & 0 & 0 & 0 & 0 & 0 & 1 
\end{array}\right)
\qquad
\mbox{and}
\qquad
\vec\theta_n=2\left(\begin{array}{cccccccc}
\overline{Y} \\ \overline{XY} \\ \overline{Y^2} \\ \overline{X^2Y^2} \\ \overline{X} \\ \overline{X^2} \\ \overline{X^3} \\ \overline{X^4} \end{array}\right),
$$
which respectively estimate
{\small $$
\Gamma_0=2\left(\begin{array}{cccccccc}
1 & \E(X) & 0 & 0 & 0 & 0 & 0 & 0  \\
\E(X) & \E(X^2) & 0 & 0 & 0 & 0 & 0 & 0  \\
0 & 0 & 1 & \E(X^2) & 0 & 0 & 0 & 0 \\
0 & 0 & \E(X^2) & \E(X^4) & 0 & 0 & 0 & 0  \\
0 & 0 & 0 & 0 & 1 & 0 & 0 & 0  \\
0 & 0 & 0 & 0 & 0 & 1 & 0 & 0 \\
0 & 0 & 0 & 0 & 0 & 0 & 1 & 0 \\
0 & 0 & 0 & 0 & 0 & 0 & 0 & 1 
\end{array}\right)
\mbox{ and }
\vec\theta_0=2\left(\begin{array}{c}
\E(Y) \\ \E(XY) \\ \E(Y^2)\\ \E(X^2Y^2)\\ \E(X)\\ \E(X^2)\\ \E(X^3)\\ \E(X^4) 
\end{array} \right).
$$} 
Then, proceeding as for the estimators based on~\eqref{varphi_lambda}, we have, provided the matrices $\Gamma_n$ and $\Gamma_0$ are invertible, that $\vec\gamma_n=\Gamma_n^{-1}\vec\theta_n$ and $\vec \gamma_0 = \Gamma_0^{-1}\vec\theta_0$. In practice, $\gamma_{n,1}$ and $\gamma_{n,2}$ (resp.\  $\gamma_{n,3}$ and $\gamma_{n,4}$) merely follow from the ordinary least-squares linear regression of $Y$ on $X$ (resp.\ $Y^2$ on $X^2$), while $\gamma_{n,4+i} = \overline{X^i}$ for $i \in \{1,\dots,4\}$.
 
To obtain an estimator of $(\alpha_0,\beta_0,\pi_0)$, we immediately have from the second term in~\eqref{varphi_gamma} corresponding to the linear regression of $Y^2$ on $X^2$ that
$$
\gamma_{0,4} = \frac{\cov(X^2,Y^2)}{\var(X^2)} =  \frac{\cov(X^2,Y^2)}{\gamma_{0,8} - \gamma_{0,6}^2},
$$
where the second equality comes from the fact that $\gamma_{0,4+i} = \E(X^i)$ for $i \in \{1,\dots,4\}$. Now, using~(\ref{Y}), we obtain
\begin{multline*}
\cov(X^2,Y^2) = \cov [ X^2, \{(1-Z) \varepsilon^* + Z (\alpha_0 + \beta_0 X + \varepsilon) \}^2 ] \\ = \pi_0 \beta_0^2 \var(X^2) + 2 \pi_0 \alpha_0 \beta_0 \cov(X^2,X).
\end{multline*}
From the first term in~\eqref{varphi_gamma} corresponding to the linear regression of $Y$ on $X$ and~\eqref{M1}, we have that $\gamma_{0,1} = \pi_0 \alpha_0$ and $\gamma_{0,2} = \pi_0 \beta_0$. Combining these with the previous display, we obtain
$$
\cov(X^2,Y^2) = \gamma_{0,2} \beta_0 (\gamma_{0,8} - \gamma_{0,6}^2) + 2 \gamma_{0,1} \beta_0 (\gamma_{0,7} - \gamma_{0,5} \gamma_{0,6}).
$$
This leads to the following estimator of $(\alpha_0,\beta_0,\pi_0)$:
\begin{align*}
\beta_n & = g^\beta(\vec\gamma_n)  = \frac{\gamma_{n,4}}{\gamma_{n,2} + 2\gamma_{n,1}(\gamma_{n,7}-\gamma_{n,5}\gamma_{n,6})/(\gamma_{n,8}-\gamma_{n,6}^2)},\\
\pi_n & = g^\pi(\vec\gamma_n)  =\frac{\gamma_{n,2}}{\beta_n}, \\
\alpha_n & =  g^\alpha(\vec\gamma_n) = \frac{\gamma_{n,1}}{\pi_n}.
\end{align*}
As we continue, the subsets of $\R^8$ on which the functions $g^\alpha$, $g^\beta$ and $g^\pi$ exist and are differentiable will be denoted by $\DD^\alpha$, $\DD^\beta$ and $\DD^\pi$, respectively, and $\DD^{\alpha,\beta,\pi}$ will stand for $\DD^\alpha \cap \DD^\beta \cap \DD^\pi$. 

To derive the asymptotic behavior of $(\alpha_n,\beta_n,\pi_n) = (g^\alpha(\vec\gamma_n),g^\beta(\vec\gamma_n),g^\pi(\vec\gamma_n))$, we consider the following assumptions:
\begin{enumerate}[{A}1.]
\item (i)~$X$ has a finite fourth order moment; (ii)~$X$ has a finite eighth order moment. 
\item the variances of $X$ and $X^2$ are strictly positive and finite.

\end{enumerate}
Clearly, Assumption A1~(ii) implies Assumption A1~(i), and Assumption A2 implies that the matrix $\Gamma_0$ defined above is invertible.

The following result, proved in Appendix~\ref{proof_euclidean}, characterizes the asymptotic behavior of the estimator $(\alpha_n,\beta_n,\pi_n)$.

\begin{prop}
\label{euclidean} 
Assume that $\vec \gamma_0 \in \DD^{\alpha,\beta,\pi}$.  
\begin{enumerate}[(i)]
\item Under Assumptions A1~(i) and A2, $(\alpha_n,\beta_n,\pi_n) \as (\alpha_0,\beta_0,\pi_0)$.
\item Suppose that Assumptions A1~(ii) and A2 are satisfied and let $\Psi_{\vec\gamma}$ be the 3 by 8 matrix defined by
$$
\Psi_{\vec\gamma}=\left(
\begin{array}{ccc}
\frac{\partial g^\alpha}{\partial \gamma_1} & \cdots & \frac{\partial g^\alpha}{\partial \gamma_8} \\
\\
\frac{\partial g^\beta}{\partial \gamma_1} & \cdots & \frac{\partial g^\beta}{\partial \gamma_8} \\
\\
\frac{\partial g^\pi}{\partial \gamma_1} & \cdots & \frac{\partial g^\pi}{\partial \gamma_8} 
\end{array}\right)(\vec \gamma), \qquad \vec \gamma \in \DD^{\alpha,\beta,\pi}.
$$
Then,
$$
\sqrt{n}(\alpha_n-\alpha_0,\beta_n-\beta_0,\pi_n-\pi_0) =  - \G_n (\Psi_{\vec \gamma_0} \Gamma_0^{-1}\dot \varphi_{\vec\gamma_0}) + o_P(1),
$$
where $\G_n = \sqrt{n} (\P_n - P)$. As a consequence, $\sqrt{n}(\alpha_n-\alpha_0,\beta_n-\beta_0,\pi_n-\pi_0)$ converges in distribution to a centered multivariate normal random vector with covariance matrix $\Sigma =\Psi_{\vec \gamma_0} \Gamma_0^{-1}  P (\dot \varphi_{\vec\gamma_0}  \dot \varphi_{\vec\gamma_0}^\top ) \Gamma_0^{-1} \Psi_{\vec \gamma_0}^\top$, which can be consistently estimated by $\Sigma_n = \Psi_{\vec \gamma_n} \Gamma_n^{-1} \P_n ( \dot \varphi_{\vec\gamma_n}  \dot \varphi_{\vec\gamma_n}^\top ) \Gamma_n^{-1} \Psi_{\vec \gamma_n}^\top$ in the sense that $\Sigma_n \as \Sigma$.
\end{enumerate}
\end{prop}

An immediate consequence of the previous result is that large-sample standard errors of $\alpha_n$, $\beta_n$ and~$\pi_n$ are given by the square root of the diagonal elements of the matrix $n^{-1} \Sigma_n$. The finite-sample performance of these estimators is investigated in Section~\ref{mc} and they are used in the illustrations of Section~\ref{illus}. 

\subsection{Estimation of the functional parameter}

To estimate the unknown c.d.f.\ $F$ of $\varepsilon$, it is natural to start from~(\ref{F}). For a known $\vec \eta=(\alpha,\beta) \in \R^2$, the term $J(\cdot,\vec \eta)$ defined in~(\ref{J}) may be estimated by the empirical c.d.f.\ of the random sample $(Y_i-\alpha-\beta X_i)_{1\leq i \leq n}$, i.e., 
$$
J_n(t,\vec \eta)= \frac{1}{n} \sum_{i=1}^n \1(Y_i-\alpha-\beta X_i\leq t), \qquad t \in \R.
$$
Similarly, since $F^*$ (the c.d.f.\ of $\varepsilon^*$) is known, a natural estimator of the term $K(t,\vec \eta)$ defined in~(\ref{K}) is given by the empirical mean of the random sample $\{F^*(t+\alpha+\beta X_i)\}_{1 \leq i \leq n}$, i.e.,
$$
K_n(t,\vec \eta)= \frac{1}{n} \sum_{i=1}^n F^*(t+\alpha+\beta X_i), \qquad t \in \R.
$$
To obtain estimators of $J(\cdot,\vec \eta_0)$ and $K(\cdot,\vec \eta_0)$, it is then natural to consider the plug-in estimators $J_n(\cdot,\vec \eta_n)$ and $K_n(\cdot,\vec \eta_n)$, respectively, based on the estimator $\vec \eta_n = (\alpha_n,\beta_n) = (g^\alpha,g^\beta)(\vec \gamma_n)$ of $\vec \eta_0$ proposed in the previous subsection. 

We shall therefore consider the following nonparametric estimator of $F$~:
\begin{equation}
\label{Fn}
F_n(t)=\frac{1}{\pi_n} \left\{ J_n(t,\vec \eta_n)- (1-\pi_n) K_n(t,\vec \eta_n) \right\}, \qquad t \in \R.
\end{equation}

Note that $F_n$ is not necessarily a c.d.f.\ as it is not necessarily increasing and can be smaller than zero or greater than one. In practice, we shall consider the partially corrected estimator $(F_n \vee 0) \wedge 1$, where $\vee$ and $\wedge$ denote the maximum and minimum, respectively.

To derive the asymptotic behavior of the previous estimator, we consider the following additional assumptions on the p.d.f.s $f^*$ and $f$ of $\varepsilon^*$ and $\varepsilon$, respectively:
\begin{enumerate}[{A}3.]

\item~(i) $f^*$ and $f$ exist and are bounded on $\R$;~(ii) $(f^*)'$ and $f'$ exist and are bounded on $\R$.

\end{enumerate}

Before stating one of our main results, let us first define some additional notation. Let $\FF^J$ and $\FF^K$ be two classes of measurable functions from $\R^2$ to $\R$ defined respectively by
$$
\FF^J = \left\{(x,y) \mapsto \psi_{t,\vec \eta}^J(x,y) = \1(y - \alpha - \beta x \leq t) : t \in \R, \vec \eta = (\alpha,\beta) \in \R^2\right\}
$$
and
$$
\FF^K = \left\{(x,y) \mapsto \psi_{t,\vec \eta}^K(x,y) = F^*(t + \alpha + \beta x) : t \in \R, \vec \eta = (\alpha,\beta) \in \R^2\right\}.
$$
Furthermore, let $\DD^{\alpha,\beta,\pi}_{\vec \gamma_0}$ be a bounded subset of $\DD^{\alpha,\beta,\pi}$ containing $\vec \gamma_0$, and let $\FF^{\alpha,\beta,\pi}$ be the class of measurable functions from $\R^2$ to $\R^3$ defined by
\begin{multline*}
\FF^{\alpha,\beta,\pi} = \left\{(x,y) \mapsto - \Psi_{\vec \gamma} \Gamma_0^{-1}\dot \varphi_{\vec\gamma}(x,y) \right. \\ \left. = \left( \psi_{\vec \gamma}^\alpha(x,y) , \psi_{\vec \gamma}^\beta(x,y), \psi_{\vec \gamma}^\pi(x,y) \right) : \vec \gamma \in \DD^{\alpha,\beta,\pi}_{\vec \gamma_0} \right\}.
\end{multline*}
With the previous notation, notice that, for any $t \in \R$,
$$
\sqrt{n} \{J_n(t,\vec \eta_0) - J(t,\vec \eta_0) \} = \G_n \psi_{t,\vec \eta_0}^J \quad \mbox{and} \quad \sqrt{n} \{K_n(t,\vec \eta_0) - K(t,\vec \eta_0) \} = \G_n \psi_{t,\vec \eta_0}^K,
$$ 
and that, under Assumptions A1~(ii) and A2, Proposition~\ref{euclidean} states that 
$$
\sqrt{n} \left( \alpha_n - \alpha_0 , \beta_n - \beta_0, \pi_n - \pi_0  \right) = \G_n \left(  \psi_{\vec \gamma_0}^\alpha , \psi_{\vec \gamma_0}^\beta , \psi_{\vec \gamma_0}^\pi \right) + o_P(1).
$$ 

Next, for any $\vec \gamma \in \DD^{\alpha,\beta,\pi}_{\vec \gamma_0}$, let
\begin{equation}
\label{psiF}
\psi_{t,\vec \gamma}^F =  \frac{1}{\pi} \psi_{t,\vec \eta}^J + f(t) \psi_{\vec \gamma}^\alpha + f(t) \E(X) \psi_{\vec \gamma}^\beta  - \frac{1 - \pi}{\pi}  \psi_{t,\vec \eta}^K + \frac{P \psi_{t,\vec \eta}^K  - P \psi_{t,\vec \eta}^J}{\pi^2} \psi_{\vec \gamma}^\pi,
\end{equation}
with $\vec \eta = (\alpha,\beta) = (g^\alpha,g^\beta)(\vec \gamma)$ and $\pi = g^\pi(\vec \gamma)$.

The following result, proved in Appendix~\ref{proof_functional}, gives the weak limit of the empirical process~$\sqrt{n}(F_n - F)$.

\begin{prop}
\label{functional}
Assume that $\vec \gamma_0 \in \DD^{\alpha,\beta,\pi}$ and that Assumptions A1, A2 and A3 hold. Then, for any $t \in \R$,
$$
\sqrt{n} \{ F_n(t) - F(t) \} = \G_n \psi_{t,\vec \gamma_0}^F + Q_{n,t},
$$
where $\sup_{t \in \R} |Q_{n,t}| = o_P(1)$ and the empirical process $t \mapsto \G_n \psi_{t,\vec \gamma_0}^F$ converges weakly to $t \mapsto \G \psi_{t,\vec \gamma_0}^F$ in $\ell^\infty(\overline{\R})$ with $\G$ a $P$-Brownian bridge.
\end{prop}

Let us now discuss the estimation of the p.d.f.\ $f$ of $\varepsilon$. Starting from~(\ref{F}) and after differentiation, it seems sensible to estimate $\E\left\{f^*(t+\alpha_0+\beta_0 X)\right\}$, $t \in \R$, by the empirical mean of the observable sample $\{f^*(t+\alpha_n+\beta_n X_i)\}_{1 \leq i \leq n}$. Hence, a natural estimator of $f$ can be defined, for any $t \in \R$, by
\begin{multline}
\label{fn}
f_n(t)=\frac{1}{\pi_n} \left\{ \frac{1}{nh_n} \sum_{i=1}^n \kappa \left(\frac{t-Y_i+\alpha_n+\beta_n X_i}{h_n}\right) \right. \\ \left. -\frac{(1-\pi_n)}{n} \sum_{i=1}^n f^*(t+\alpha_n+\beta_n X_i) \right\},
\end{multline}
where $\kappa$ is a kernel function on $\R$ and $(h_n)_{n\geq 1}$ is a sequence of bandwidths converging to zero. 

In the same way that $F_n$ is not necessarily a c.d.f., $f_n$ is not necessarily a p.d.f. In practice, we shall use the partially corrected estimator $f_n \vee 0$. A fully corrected estimator (so that, additionally, the estimated density integrates to one) can be obtained as explained in~\cite{GlaHjoUsh03}.

Consider the following additional assumptions on $(h_n)_{n\geq 1}$, $\kappa$ and $f^*$~:
\begin{enumerate}[{A}4.]
\item (i)~$h_n=cn^{-\alpha}$ with $\alpha\in(0,1/2)$ and $c>0$ a constant;~(ii)~$\kappa$ is a p.d.f.\ with bounded variations on $\R$ and a finite first order moment;~(iii) the p.d.f.\ $f^*$ has bounded variations on $\R$.
\end{enumerate}

The following result is proved in Appendix~\ref{proof_pdf}.

\begin{prop} 
\label{pdf}
If $\vec\gamma_0\in{\cal D}^{\alpha,\beta,\pi}$, and under Assumptions A1~(i), A2, A3 and A4, 
$$
\sup_{t \in \R} |f_n(t) - f(t)| \as 0.
$$
\end{prop}

Finally, note that, in all our numerical experiments, the kernel part of $f_n$ was computed using the {\tt ks} R package~\cite{ks} in which the univariate plug-in selector proposed in~\cite{WanJon94} was used for the bandwidth $h_n$.

\subsection{A weighted bootstrap with application to confidence bands for $F$}
\label{bootstrap}

In applications, it may be of interest to carry out inference on $F$. The result stated in this section can be used for that purpose. It is based on the unconditional multiplier central limit theorem for empirical processes~\cite[see e.g.][Theorem 10.1 and Corollary 10.3]{Kos08} and can be used to obtain approximate independent copies of $\sqrt{n}(F_n - F)$. 

Given i.i.d.\ mean 0 variance 1 random variables $\xi_1,\dots,\xi_n$ with $\int_0^\infty \{ \Pr(|\xi_1| > x) \}^{1/2} \dd x < \infty$, and independent of the random sample $(X_i,Y_i)_{1\leq i\leq n}$, let
$$
\G_n' = \frac{1}{\sqrt{n}} \sum_{i=1}^n (\xi_i - \bar \xi) \delta_{X_i,Y_i},
$$
where $\bar \xi = n^{-1} \sum_{i=1}^n \xi_i$. Also, let $\left( \hat \psi_{\vec \gamma_n}^\alpha,\hat \psi_{\vec \gamma_n}^\beta,\hat \psi_{\vec \gamma_n}^\pi \right) = -\Psi_{\vec \gamma_n} \Gamma_n^{-1}\dot \varphi_{\vec\gamma_n}$ and, for any $t \in \R$, let
\begin{equation}
\label{hatpsiF}
\hat \psi_{t,\vec \gamma_n}^F =  \frac{1}{\pi_n} \psi_{t,\vec \eta_n}^J + f_n(t) \hat \psi_{\vec \gamma_n}^\alpha + f_n(t) \bar X \hat \psi_{\vec \gamma_n}^\beta  - \frac{1 - \pi_n}{\pi_n}  \psi_{t,\vec \eta_n}^K + \frac{\P_n \psi_{t,\vec \eta_n}^K  - \P_n \psi_{t,\vec \eta_n}^J}{\pi_n^2} \hat \psi_{\vec \gamma_n}^\pi
\end{equation}
be an estimated version of the influence function $\psi_{t,\vec \gamma_0}^F$ arising in Proposition~\ref{functional}, where $\vec \eta_n = (\alpha_n,\beta_n) = (g^\alpha,g^\beta)(\vec \gamma_n)$ and $\pi_n = g^\pi(\vec \gamma_n)$. 

The following proposition, proved in Appendix~\ref{proof_mult}, suggests, when $n$ is large, to interpret $t \mapsto \G_n' \hat \psi_{t,\vec \gamma_n}^F$ as an independent copy of $\sqrt{n}(F_n - F)$. 

\begin{prop}
\label{mult}
Assume that $\vec \gamma_0 \in \DD^{\alpha,\beta,\pi}$, and that Assumptions A1, A2, A3 and A4 hold. Then, the process $( t \mapsto \G_n  \psi_{t,\vec \gamma_0}^F, t \mapsto \G_n' \hat \psi_{t,\vec \gamma_n}^F  )$ converges weakly to $( t \mapsto  \G  \psi_{t,\vec \gamma_0}^F, t \mapsto  \G' \psi_{t,\vec \gamma_0}^F )$ in $\{ \ell^\infty (\overline{\R}) \}^2$, where $t \mapsto \G' \psi_{t,\vec \gamma_0}^F$ is an independent copy of $t \mapsto \G  \psi_{t,\vec \gamma_0}^F$.
\end{prop}

Let us now explain how the latter result can be used in practice to obtain an approximate confidence band for $F$. Let $N$ be a large integer and let $\xi_i^{(j)}$, $i \in \{1,\dots,n\}$, $j \in \{1,\dots,N\}$, be i.i.d.\ random variables with mean 0, variance 1, satisfying $\int_0^\infty \{ \Pr(|\xi_i^{(j)}| > x) \}^{1/2} \dd x < \infty$, and independent of the data $(X_i,Y_i)_{1\leq i\leq n}$. For any $j \in \{1,\dots,N\}$, let $\G_n^{(j)} = n^{-1/2} \sum_{i=1}^n (\xi_i^{(j)} - \bar \xi^{(j)}) \delta_{X_i,Y_i}$, where $\bar \xi^{(j)} = n^{-1} \sum_{i=1}^n \xi_i^{(j)}$. Then, a consequence of Propositions~\ref{functional} and~\ref{mult} is that
\begin{multline*}
\left( \sqrt{n}(F_n - F), t \mapsto \G_n^{(1)} \hat \psi_{t,\vec \gamma_n}^F,\dots, t \mapsto \G_n^{(N)} \hat \psi_{t,\vec \gamma_n}^F \right) \\ \leadsto \left( t \mapsto  \G  \psi_{t,\vec \gamma_0}^F, t \mapsto  \G^{(1)} \psi_{t,\vec \gamma_0}^F, \dots, t \mapsto  \G^{(N)} \psi_{t,\vec \gamma_0}^F \right)
\end{multline*}
in $\{ \ell^\infty (\overline{\R}) \}^{N+1}$, where $\G^{(1)},\dots,\G^{(N)}$ are independent copies of the $P$-Brownian bridge $\G$. From the continuous mapping theorem, it follows that 
\begin{multline*}
\left( \sup_{t \in \R} |\sqrt{n}(F_n - F)|, \sup_{t \in \R} |\G_n^{(1)} \hat \psi_{t,\vec \gamma_n}^F |,\dots, \sup_{t \in \R} |\G_n^{(N)} \hat \psi_{t,\vec \gamma_n}^F | \right) \\ \leadsto \left( \sup_{t \in \R} | \G  \psi_{t,\vec \gamma_0}^F |, \sup_{t \in \R} | \G^{(1)} \psi_{t,\vec \gamma_0}^F |, \dots, \sup_{t \in \R} | \G^{(N)} \psi_{t,\vec \gamma_0}^F | \right)
\end{multline*}
in $[0,\infty)^{N+1}$. The previous result suggests to estimate quantiles of $\sup_{t \in \R} |\sqrt{n}(F_n - F)|$ using the generalized inverse of the empirical c.d.f.\
\begin{equation}
\label{GnN}
G_{n,N}(x) = \frac{1}{N} \sum_{j=1}^N \1 \left\{ \sup_{t \in \R} |\G_n^{(j)} \hat \psi_{t,\vec \gamma_n}^F | \leq x \right\}.
\end{equation}
A large-sample confidence band of level $1-p$ for $F$ is thus given by $F_n \pm G_{n,N}^{-1}(1-p) / \sqrt{n}$. Examples of such confidence bands are given in Figures~\ref{WOMOSO},~\ref{tone} and~\ref{aphids}, and the finite-sample properties of the above construction are empirically investigated in Section~\ref{mc}. Note that in all our numerical experiments, the random variables $\xi_i^{(j)}$ were taken from the standard normal distribution, and that the supremum in the previous display was replaced by a maximum over 100 points $U_1,\dots,U_{100}$ uniformly spaced over the interval $[\min_{1 \leq i \leq n} (Y_i - \alpha_n - \beta_n X_i), \max_{1 \leq i \leq n} (Y_i - \alpha_n - \beta_n X_i)]$. 

Finally, notice that Proposition~\ref{functional} implies that, for any fixed $t \in \R$, 
the random variable $\G_n  \psi_{t,\vec \gamma_0}^F$ converges in distribution to $\G  \psi_{t,\vec \gamma_0}^F$. This suggests to estimate the variance of $\G  \psi_{t,\vec \gamma_0}^F$ as the variance of $\G_n  \psi_{t,\vec \gamma_0}^F$, which is equal to $\var \{ \psi_{t,\vec \gamma_0}^F(X,Y) \} = P (\psi_{t,\vec \gamma_0}^F)^2 - (P \psi_{t,\vec \gamma_0}^F)^2$. Should $\vec \gamma_0$ be known, a natural estimate of the latter would be the empirical variance of the random sample $\{\psi_{t,\vec \gamma_0}^F(X_i,Y_i)\}_{1 \leq i \leq n}$. As $\vec \gamma_0$ is unknown, the sample of ``pseudo-observations'' $\{\hat \psi_{t,\vec \gamma_n}^F(X_i,Y_i)\}_{1 \leq i \leq n}$ can be used instead. This suggests to estimate the standard error of $F_n(t)$ as 
\begin{equation}
\label{std}
n^{-1/2} \{ \P_n (\hat \psi_{t,\vec \gamma_n}^F)^2 - (\P_n \hat \psi_{t,\vec \gamma_n}^F)^2 \}^{1/2}.
\end{equation}
The finite-sample performance of this estimator is investigated in Section~\ref{mc} for several values of~$t$.

\section{Monte Carlo experiments}
\label{mc}

A large number of Monte Carlo experiments was carried out to investigate the influence on the estimators of various factors such as the degree of overlap of the mixed populations, the proportion of the unknown component $\pi_0$, or the shape of the noise $\varepsilon$ involved in the unknown regression model. Starting from~(\ref{model}), the following generic data generating models were considered:
\begin{align*}
\mbox{WO}:& \, \varepsilon^* \sim \NN(0,1), \, (\alpha_0,\beta_0)=(2,1),  \, X \sim \NN(2,3^2),  \, \E(\varepsilon^2)=1, \\
\mbox{MO}:& \, \varepsilon^* \sim \NN(0,1), \,  (\alpha_0,\beta_0)=(2,1),  \, X\sim \NN(2,3^2),  \, \E(\varepsilon^2)=4, \\
\mbox{SO}:&  \, \varepsilon^* \sim \NN(0,1), \, (\alpha_0,\beta_0)=(1,0.5),  \, X \sim \NN(1,2^2),  \, \E(\varepsilon^2)=4.
\end{align*}
The abbreviations WO, MO and SO stand respectively for ``Weak Overlap'', ``Medium Overlap'' and ``Strong Overlap''. Three possibilities were considered for the distribution of $\varepsilon$: the centered normal (the corresponding data generating models will be abbreviated by WOn, MOn and SOn), a gamma distribution with shape parameter equal to two and rate parameter equal to a half, shifted to have mean zero (the corresponding models will be abbreviated by WOg, MOg and SOg) and a standard exponential shifted to have mean zero (the corresponding models will be abbreviated by WOe, MOe and SOe). Depending on the model they are used in, all three error distributions are scaled so that $\varepsilon$ has the desired variance. 

Examples of datasets generated from WOn, MOg and SOe with $n=500$ and $\pi_0=0.7$ are represented in the first column of graphs of Figure~\ref{WOMOSO}. The solid (resp.\ dashed) lines represent the true (resp.\ estimated) regression lines. The graphs of the second column represent, for each of WOn, MOg and SOe, the true c.d.f.\ $F$ of $\varepsilon$ (solid line) and its estimate $F_n$ (dashed line) defined in~(\ref{Fn}). The dotted lines represent approximate confidence bands of level 0.95 for $F$ computed as explained in Subsection~\ref{bootstrap} with $N=10,000$. Finally, the graphs of the third column represent, for each of WOn, MOg and SOe, the true p.d.f.\ $f$ of $\varepsilon$ (solid line) and its estimate $f_n$ (dashed line) defined in~(\ref{fn}). 

\begin{figure}[t!]
  \begin{center}
    \includegraphics*[width=1\linewidth]{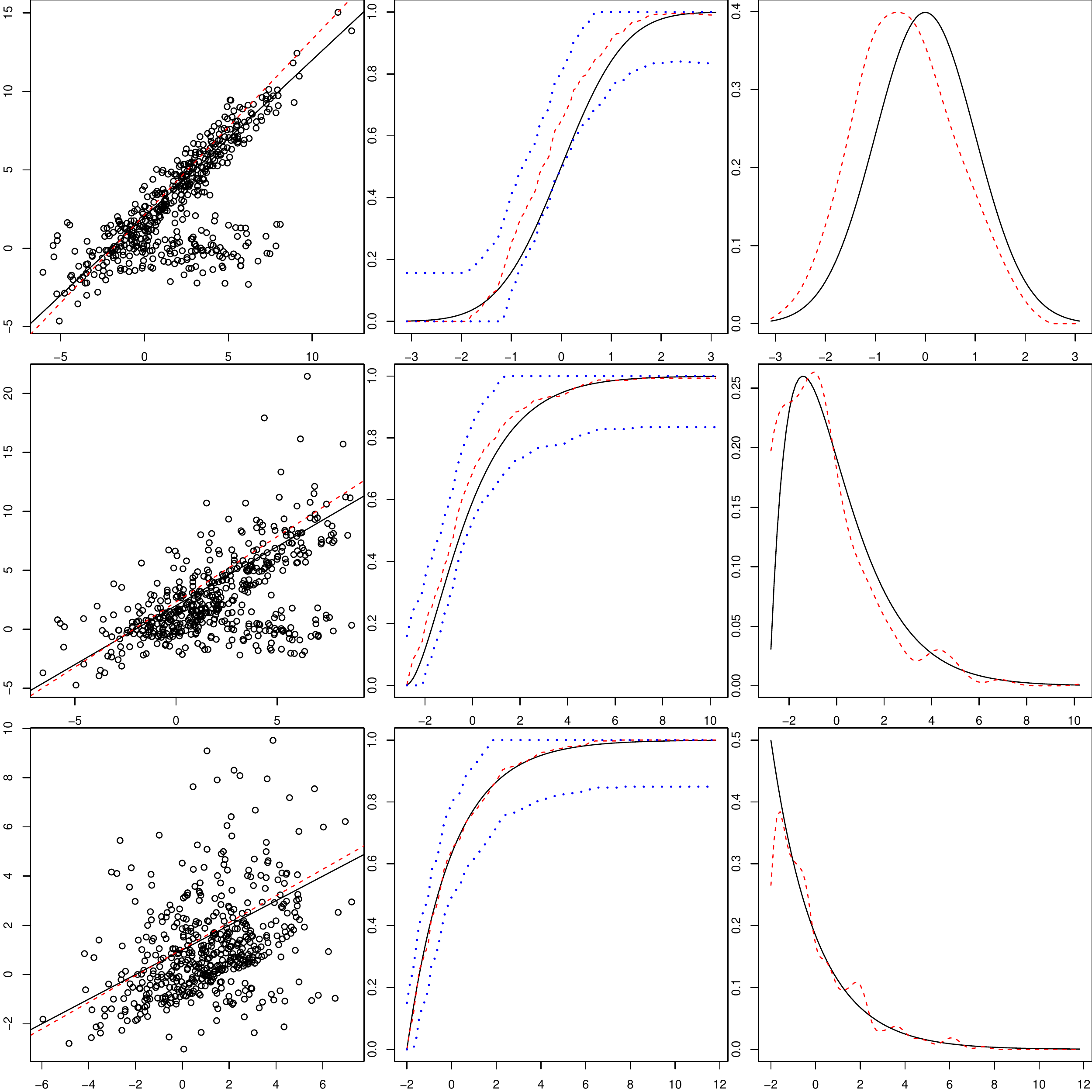}
  \end{center}
  \caption{First column, from top to bottom: datasets generated from WOn, MOg and SOe, respectively, with $n=500$ and $\pi_0=0.7$; the solid (resp.\ dashed) lines represent the true (resp.\ estimated) regression lines. Second column, from top to bottom: for WOn, MOg and SOe, respectively, the true c.d.f.\ $F$ of $\varepsilon$ (solid line) and its estimate $F_n$ (dashed line) defined in~(\ref{Fn}). The dotted lines represent approximate confidence bands of level 0.95 for $F$ computed as explained in Subsection~\ref{bootstrap} with $N=10,000$. Third column, from top to bottom: for WOn, MOg and SOe, respectively, the true p.d.f.\ $f$ of $\varepsilon$ (solid line) and its estimate $f_n$ defined in~(\ref{fn}) (dashed line).}
  \label{WOMOSO}
\end{figure}

For each of the three groups of data generating models, \{WOn, MOn, SOn\}, \{WOg, MOg, SOg\} and \{WOe, MOe, SOe\}, the values 0.4 and 0.7 were considered for $\pi_0$, and the values 100, 300, 1000 and 5000 were considered for $n$. For each of the nine data generating scenarios, each value of $\pi_0$, and each value of $n$, $M=1000$ random samples were generated. Tables~\ref{normnoise},~\ref{gammanoise} and~\ref{expnoise} report the number $m$ of samples out of $M$ for which $\pi_n \not \in (0,1]$, as well as the estimated bias and standard deviation of $\alpha_n$, $\beta_n$, $\pi_n$, $F_n\{F^{-1}(0.1)\}$, $F_n\{F^{-1}(0.5)\}$ and $F_n\{F^{-1}(0.9)\}$ computed from the $M-m$ valid estimates.




\input{normnoise}

\input{gammanoise}

\input{expnoise}

A first general comment concerning the results reported in Tables~\ref{normnoise},~\ref{gammanoise} and~\ref{expnoise} is that the number $m$ of samples for which $\pi_n \not \in (0,1]$ is the highest for the SO scenarios followed by the MO scenarios and then the WO scenarios. Also, for a fixed amount of overlap between the two mixed populations, it is when the distribution of $\varepsilon$ is exponential that $m$ tends to be the highest followed by the gamma and the normal cases. Hence, as expected, the SO scenarios are the hardest and, for a given degree of overlap, the most difficult problems are those involving exponential errors for the unknown regression component.

{\it Influence of the shape of the p.d.f.\ of $\varepsilon$.}  A surprising result, when observing Tables~\ref{normnoise},~\ref{gammanoise} and~\ref{expnoise}, is that the nature of the distribution of $\varepsilon$ appears to have very little influence on the performance of the estimators $\alpha_n$, $\beta_n$ and $\pi_n$. Under weak and moderate overlap in particular, the estimated bias and standard deviations of the estimators are almost unaffected by the distribution of the error of the unknown component. 

{\it The effect of the degree of overlap.} As expected, the performance of the estimators $\alpha_n$, $\beta_n$ and $\pi_n$ is strongly affected by the degree of overlap. Notice however that the results obtained under the WO and MO data generating scenarios are rather comparable, while the performance of the estimators gets significantly worse when switching to the SO scenarios, especially for $\pi_n$.
Notice also that, overall, the biases of $\alpha_n$ and $\beta_n$ are negative under WO and MO and positive under SO, while, for all the scenarios under consideration, $\pi_n$ tends to have a positive bias.

{\it The influence of $\pi_0$.} For a given degree of overlap and sample size, the parameter that seems to affect the most the performance of the estimators   is the proportion $\pi_0$ of the unknown component. On one hand, the number of samples for which $\pi_n \notin (0,1]$ is lower for $\pi_0=0.4$ than for $\pi_0=0.7$. On the other hand, when considering the samples for which $\pi_n \in (0,1]$, the finite-sample behavior of $\alpha_n$ and $\beta_n$ improves very clearly when $\pi_0$ switches from $0.4$ to $0.7$. 

{\it Performance of the functional estimator.}  The study of $F_n\{F^{-1}(p)\}$ for $p \in \{0.1, 0.5, 0.9\}$ clearly shows that, for a given degree of overlap between the two mixed populations, the performance of the functional estimator is the best when the distribution of $\varepsilon$ is normal followed by the gamma and the exponential settings. In addition, it appears that $F_n\{F^{-1}(p)\}$, $p \in \{0.1, 0.5\}$, behaves the best under the MO scenarios, and that, somehow surprisingly, $F_n\{F^{-1}(0.9)\}$ achieves its best results under the SO scenarios.

{\it Asymptotics.} The results reported in Tables~\ref{normnoise},~\ref{gammanoise} and~\ref{expnoise} are in accordance with the asymptotic theory stated in the previous section. In particular, as expected, the estimated biases and standard deviations of all the estimators tend to zero as $n$ increases. Notice for instance that under SOg and SOe with $\pi_0 = 0.4$ (two of the most difficult scenarios), the estimated standard deviation of $\alpha_n$ is greater than 7 for $n=100$, drops below $0.7$ for $n=300$, and becomes very reasonable for $n=1000$ and $5000$.

{\it Comparison with the method proposed in~\cite{Van13}.} The results reported in Table~\ref{normnoise} for models WOn, MOn and SOn, and for $n \in \{100,300\}$, can be directly compared with those reported in~\cite[Table 2]{Van13}. The scenarios with gamma and exponential errors considered in this work have however no analogue in~\cite{Van13} as the method therein was derived under zero-symmetry assumptions for the errors. A comparison of Table~\ref{normnoise} with Table~2 in~\cite{Van13} reveals that the standard deviations of our estimators of $\alpha_0$, $\beta_0$ and $\pi_0$ are between 1.5 and 3 times larger, while the two sets of estimators are rather comparable in terms of bias. It is however important to recall that the results reported in~\cite{Van13} were obtained after a careful adjustment of the tuning parameters of the estimation method while the approach derived in this work is free of tuning parameters. Indeed, in the Monte Carlo experiments reported in~\cite{Van13}, the underlying gradient optimization method is initialized at the true value of the parameter vector $(\alpha_0,\beta_0,\pi_0)$ and the choice of the weight distribution function involved in the definition of the contrast function is carefully hand-tuned to avoid numerical instability~\cite[see][Section 4.2]{Van13}. 

Let us now present the results of the Monte Carlo experiments used to investigate the finite-sample performance of the estimators of the standard errors of $\alpha_n$, $\beta_n$, $\pi_n$ and $F_n\{F^{-1}(p)\}$, $p \in \{0.1, 0.5, 0.9\}$, mentioned below Proposition~\ref{euclidean} and in~\eqref{std}, respectively. The setting is the same as previously with the exception that $n \in \{100,300,1000,5000,25000\}$. The results are partially reported in Table~\ref{stdall} which gives, for scenarios WOn, MOg and SOe and each of the aforementioned estimators, the standard deviation of the estimates multiplied by $\sqrt{n}$ and the mean of the estimated standard errors multiplied by $\sqrt{n}$. As can be seen, for all estimators and all scenarios, the standard deviation of the estimates and the mean of the estimated standard errors are always very close for $n=25,000$. The convergence to zero of the difference between these two quantities appears however slower for $F_n\{F^{-1}(p)\}$, $p \in \{0.1, 0.5, 0.9\}$, than for $\alpha_n$, $\beta_n$ and $\pi_n$, the worst results being obtained for $F_n\{F^{-1}(0.1)\}$. The results also confirm that the SO scenarios are the hardest. Notice finally that the estimated standard errors of $\alpha_n$ and $\beta_n$ seem to underestimate on average the variability of $\alpha_n$ and $\beta_n$, and that the variability of $\pi_n$ and $F_n\{F^{-1}(p)\}$, $p \in \{0.1, 0.5, 0.9\}$ appears to be underestimated on average for the WO and MO scenarios, and overestimated on average for the SO scenarios.


\input{stdall}

We end this section by an investigation of the finite-sample properties of the confidence band construction proposed in Subsection~\ref{bootstrap}. Table~\ref{cball} reports the proportion of samples for which 
$$
\max_{t \in \{U_1,\dots,U_{100}\}} | F_n(t) - F(t) | > n^{-1/2} G_{n,N}^{-1}(0.95), 
$$
where $G_{n,N}$ is defined as in~(\ref{GnN}) with $N=1000$, and $U_1,\dots,U_{100}$ are uniformly spaced over the interval $[\min_{1 \leq i \leq n} (Y_i - \alpha_n - \beta_n X_i), \max_{1 \leq i \leq n} (Y_i - \alpha_n - \beta_n X_i)]$. As could have been partly expected from the results reported in Table~\ref{stdall}, the confidence bands are too narrow on average for the WO and MO scenarios, the worse results being obtained when the error of the unknown component is exponential. The results are, overall, more satisfactory for the SO scenarios. In all cases, the estimated coverage probability appears to converge to 0.95, although the convergence appears to be slow.


\input{cball}

\section{Illustrations}
\label{illus}

We first applied the proposed method to a dataset initially reported in~\cite{Coh80} and subsequently analyzed in~\cite{DeV89} and~\cite{HunYou12}, among others. As we shall see, the model studied in this work and stated in~\eqref{model} appears as a rather natural candidate for this dataset. For other datasets for which it is less natural to assume that one of the two components is known, the derived method can be used to assess the relevance of the results of EM-type algorithms for estimating two-component mixtures of linear regressions. Two such datasets will be analyzed: the aphids dataset initially considered in~\cite{BoiSinSinTaiTur98}, and the NimbleGen high density array dataset studied in~\cite{MarMarBer08}.

\subsection{The tone dataset}

The dataset consists of $n=150$ observations $(x_i,\tilde y_i)$ where the $x_i$ are actual tones and the $\tilde y_i$ are the corresponding perceived tones by a trained musician. The detailed description of the dataset given in~\cite{HunYou12} suggests that it is natural to consider that the equation of the tilted component is $y=x$. The transformation $y_i = \tilde y_i - x_i$ was then applied to obtain a dataset $(x_i,y_i)$ that fits into the setting considered in this work. The original dataset and the transformed dataset are represented in the upper left and upper right plots of Figure~\ref{tone}, respectively.

\begin{figure}[t!]
  \begin{center}
    \includegraphics*[width=1\linewidth]{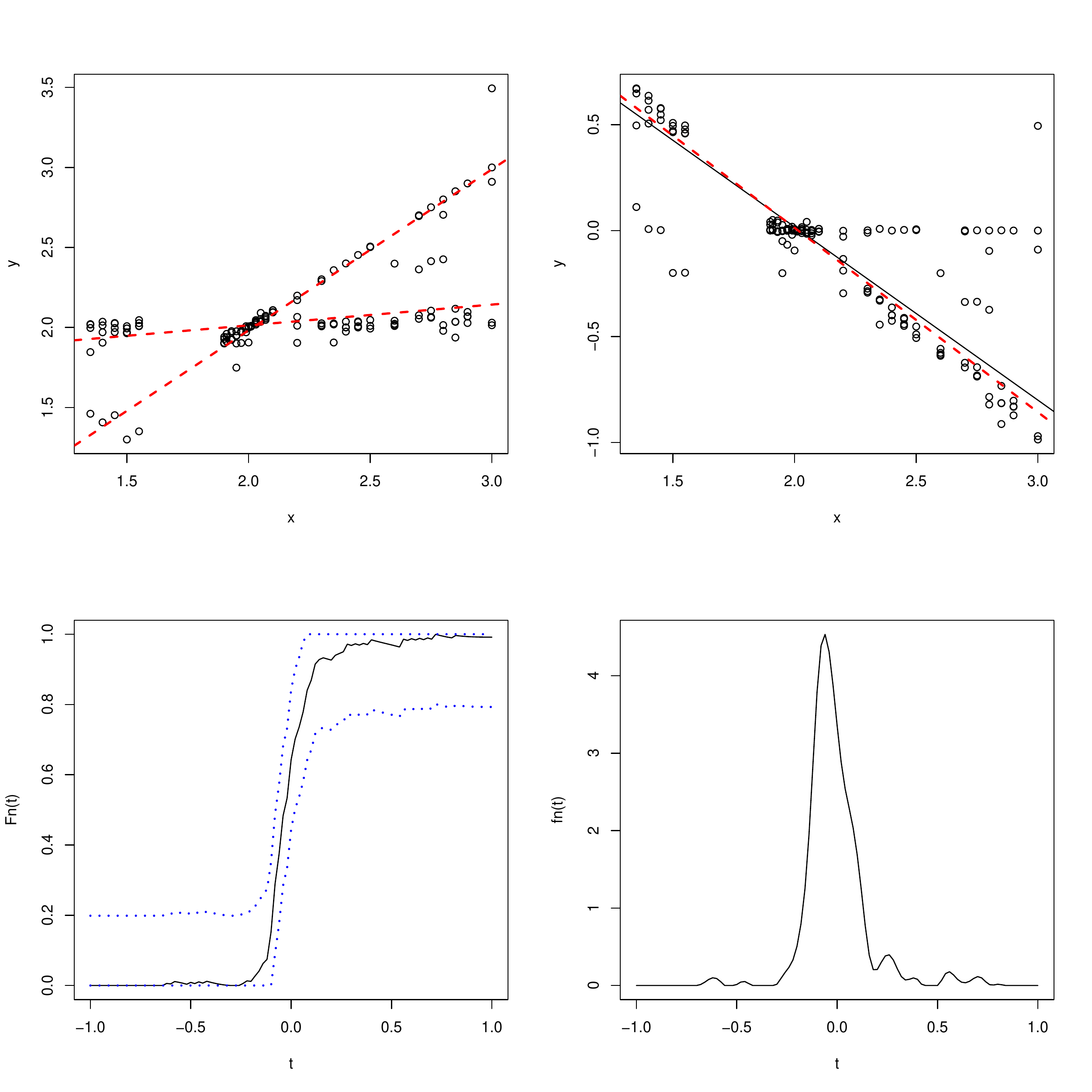}
  \end{center}
  \caption{Upper left plot: the original tone data; the dashed lines represent the regression lines obtained in~\cite{HunYou12} using a semiparametric EM-like algorithm without zero-symmetry assumptions. Upper right plot: the transformed data; the solid line represents the estimated regression line; the dashed line represents the corresponding (transformed) regression line obtained in~\cite{HunYou12}. Lower left plot: the estimate $(F_n \vee 0) \wedge 1$ (solid line) of the unknown c.d.f.\ $F$ of $\varepsilon$ as well as an approximate confidence band (dotted lines) of level 0.95 for $F$ computed as explained in Subsection~\ref{bootstrap} with $N=10,000$. Lower right plot: the estimate $f_n \vee 0$ of the unknown p.d.f.\ $f$ of $\varepsilon$.}
  \label{tone}
\end{figure}

The approach proposed in this paper was applied under the assumption that the distribution of $\varepsilon^*$ in~(\ref{model}) is normal with standard deviation 0.079. The latter value was obtained by considering the upper right plot of Figure~\ref{tone} and by computing the sample standard deviation of the $y_i$ such that $y_i \in (-0.25, 0.25)$ and $x_i < 1.75$ or $x_i > 2.25$. 

The estimate $(1.652, -0.817,  0.790)$ was obtained for the parameter vector $(\alpha_0,\beta_0,\pi_0)$ with $(0.217,  0.108, 0.104)$ as the vector of estimated standard errors. The estimated regression line is represented by a solid line in the upper right plot of Figure~\ref{tone}. The dashed line represents the corresponding (transformed) regression line obtained in~\cite{HunYou12} using a semiparametric EM-like algorithm without zero-symmetry assumptions (see Table~1 in the latter paper for more results). The estimate $(F_n \vee 0) \wedge 1$ (resp.\ $f_n \vee 0$) of the unknown c.d.f.\ $F$ (resp.\ p.d.f.\ $f$) of $\varepsilon$ is represented in the lower left (resp.\ right) plot of Figure~\ref{tone}. The dotted lines in the lower left plot represent an approximate confidence band of level 0.95 for $F$ computed as explained in Subsection~\ref{bootstrap} using $N=10,000$. Note that, from the results of the previous section, the latter is probably too narrow. Numerical integration using the R function {\tt integrate} gave $\int_{-1}^1 (f_n \vee 0) \approx 1.01$. 

\subsection{The aphids dataset}

We next considered a dataset initially analyzed in~\cite{BoiSinSinTaiTur98} and available in the {\tt mixreg} R package~\cite{mixreg}. The data were obtained from 51 experiments. Each experiment consisted of releasing a certain number of green peach aphids (flying insects) in a chamber containing 81 tobacco plants arranged in a $9 \times 9$ grid. Among these plants, 12 were infected by a certain virus and 69 were healthy. After 24 hours the chambers were fumigated to kill the aphids, and the previously healthy plants were moved and monitored to detect symptoms of infection. The number of infected plants was recorded. The dataset thus consists of $n=51$ observations $(x_i,\tilde y_i)$ where the $x_i$ are the number of released aphids and the $\tilde y_i$ are the corresponding number of infected plants. The resulting scatterplot is represented in the upper left plot of Figure~\ref{aphids}. The dashed lines represent the regression lines obtained in~\cite[Table 1]{Tur00} using a standard EM algorithm with normal errors. With the notation of~\eqref{twocomp} and the convention that $\sigma_0^*$ and $\sigma_0^{**}$ are the standard deviations of $\varepsilon^*$ and $\varepsilon^{**}$, respectively, the author obtained the estimate $(0.859, 0.002, 1.125)$ for $(\alpha_0^*,\beta_0^*,\sigma_0^*)$, the estimate $(3.47, 0.055, 3.115)$ for $(\alpha_0^{**},\beta_0^{**},\sigma_0^{**})$ and the estimate 0.5 for $\pi_0$~\cite[see also][Table 4]{YouHun10}. 

\begin{figure}[t!]
  \begin{center}
    \includegraphics*[width=1\linewidth]{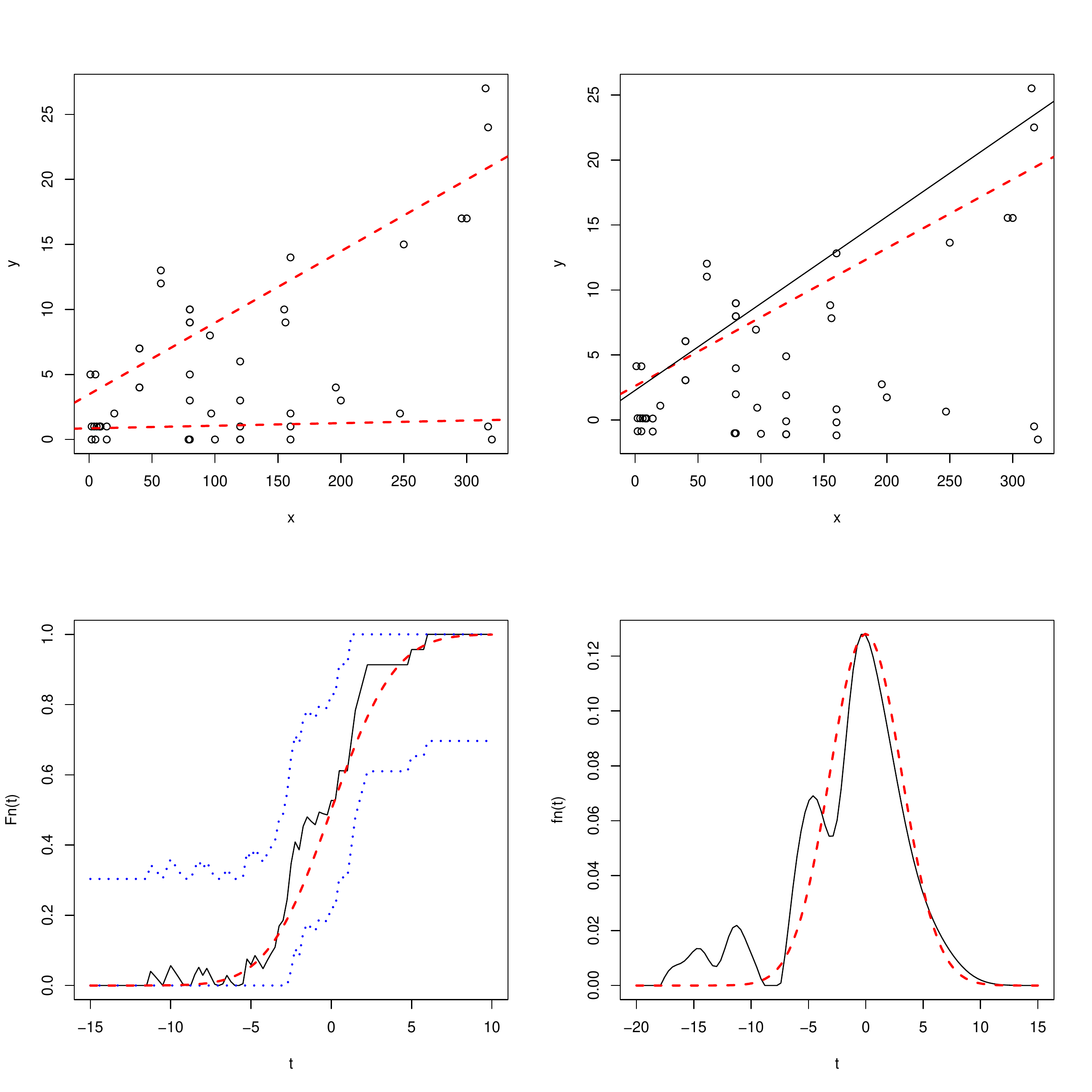}
  \end{center}
  \caption{Upper left plot: the original aphid data; the dashed lines represent the regression lines reported in~\cite[Table 1]{Tur00} and obtained using a standard EM approach with normal errors. Upper right plot: the transformed aphid data; the solid line represents the estimated regression line; the dashed line represents the corresponding (transformed) regression line obtained in~\cite{Tur00}. Lower left plot: the estimate $(F_n \vee 0) \wedge 1$ (solid line) of the unknown c.d.f.\ $F$ of $\varepsilon$, the c.d.f.\ of the parametric estimation of $\varepsilon^{**}$ obtained in~\cite{Tur00} (dashed line) as well as an approximate confidence band (dotted lines) of level 0.95 for $F$ computed as explained in Subsection~\ref{bootstrap} with $N=10,000$. Lower right plot: the estimate $f_n \vee 0$ (solid line) of the unknown p.d.f.\ $f$ of $\varepsilon$ and the p.d.f.\ of the parametric estimation of $\varepsilon^{**}$ obtained in~\cite{Tur00} (dashed line).}
  \label{aphids}
\end{figure}

To show how the semiparametric approach studied in this work could be used to assess the relevance of the results reported in~\cite{Tur00}, we arbitrarily made the assumption that the almost horizontal component in the upper left plot of Figure~\ref{aphids} was perfectly estimated, i.e., that the known component has equation $y=0.859 + 0.002x$ and that the distribution of the corresponding error is normal with standard deviation 1.125 as estimated in~\cite{Tur00}. The transformation $y_i = \tilde y_i - 0.859 - 0.002 x_i$ was then applied to obtain a dataset $(x_i,y_i)$ that fits into the setting considered in this work.  The resulting scatterplot is represented in the upper right plot of Figure~\ref{aphids}.

The estimate $(2.281, 0.067, 0.454)$ was obtained for the parameter $(\alpha_0,\beta_0,\pi_0)$ with $(2.538, 0.016, 0.120)$ as the vector of estimated standard errors. The estimated regression line is represented by a solid line in the upper right plot of Figure~\ref{aphids}. The dashed line represents the corresponding regression line obtained in~\cite{Tur00}. The estimate $(F_n \vee 0) \wedge 1$ (resp.\ $f_n \vee 0$) of the unknown c.d.f.\ $F$ (resp.\ p.d.f.\ $f$) of $\varepsilon$ is represented in the lower left (resp.\ right) plot of Figure~\ref{aphids}. The dashed curve in the lower left (resp.\ right) plot represents the c.d.f.\ (resp.\ p.d.f.) of the parametric estimation of $\varepsilon^{**}$ obtained in~\cite{Tur00}, which is normal with standard deviation equal to $\sigma_0^{**} = 3.115$. The dotted lines in that the lower left plot represent an approximate confidence band of level 0.95 for $F$ computed as explained in Subsection~\ref{bootstrap} using $N=10,000$. Note again that, from the results of the previous section, the latter is probably too narrow. Numerical integration using the R function {\tt integrate} gave $\int_{-20}^{15} (f_n \vee 0) \approx 1.07$. The results reported in Figure~\ref{aphids} show no evidence against a normal assumption for the error of the second component.

\subsection{The NimbleGen high density array dataset}

As a final application, we considered the NimbleGen high density array dataset analyzed initially in~\cite{MarMarBer08}. The dataset, produced by a two color ChIP-chip experiment, consists of $n=176,343$ observations $(x_i,\tilde y_i)$. The corresponding scatter plot is represented in the upper left plot of Figure~\ref{chipmix}. A parametric mixture of linear regressions with two unknown components was fitted to the data in~\cite{MarMarBer08} under the assumption of normal errors using a standard EM algorithm. The estimates are reported in~\cite[Section 4.4]{Van13} in the homoscedastic and heteroscedastic cases, and the regression lines obtained in the heteroscedastic case are represented by dashed lines in the upper left plot of Figure~\ref{chipmix}. As for the aphids dataset, we used the approach derived in this work to assess the relevance of the latter results. We arbitrarily considered that the component with the smallest slope was precisely estimated, i.e., that it has equation $y = 1.48 + 0.81 x$, and that the distribution of the corresponding error is normal with a standard deviation of 0.56 as obtained in~\cite{MarMarBer08} and reported in~\cite[Section 4.4]{Van13}. The transformation $y_i = \tilde y_i - (1.48 + 0.81 x_i)$ was then performed to obtain a dataset $(x_i,y_i)$ that fits into the setting considered in this work. The transformed dataset is represented in the upper right plot of Figure~\ref{chipmix}.

\begin{figure}[t!]
  \begin{center}
    \includegraphics*[width=1\linewidth]{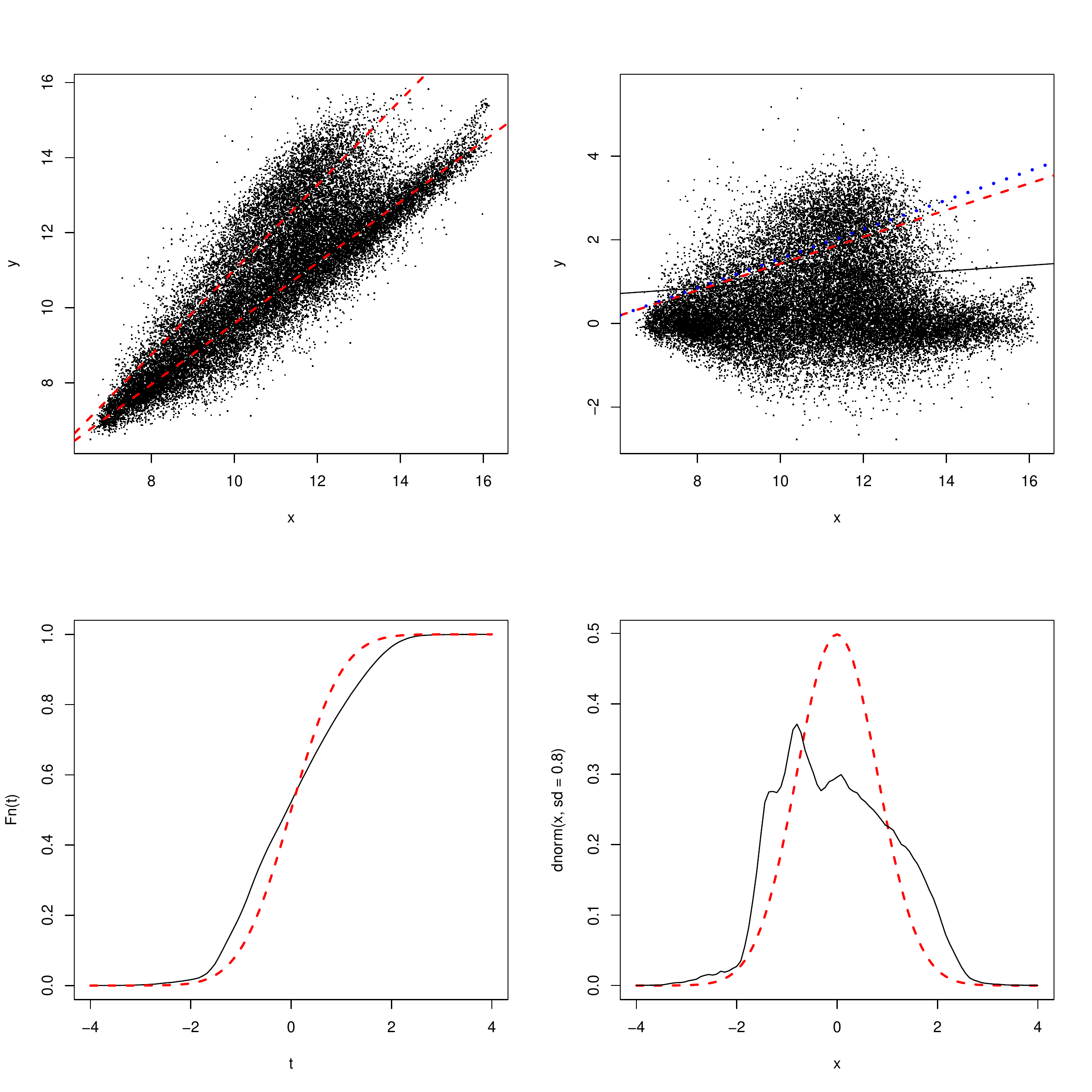}
  \end{center}
  \caption{Upper left plot: the original ChIPmix data and the regression lines obtained in~\cite{MarMarBer08} using a standard EM algorithm with normal errors. Upper right plot: the transformed data; the solid line represents the regression line estimated by the method derived in this work, while the dashed (resp.\ dotted) line is the corresponding regression line estimated in~\cite{MarMarBer08} (resp.\ in~\cite{Van13}). Lower left plot: the estimate $(F_n \vee 0) \wedge 1$ (solid line) of the unknown c.d.f.\ $F$ of $\varepsilon$ and the c.d.f.\ of the parametric estimation of the corresponding error obtained in~\cite{MarMarBer08} (dashed line). Lower right plot: the estimate $f_n \vee 0$ (solid line) of the unknown p.d.f.\ $f$ of $\varepsilon$ and the p.d.f.\ of the parametric estimation of the corresponding error obtained in~\cite{MarMarBer08} (dashed line).}
\label{chipmix}
\end{figure}

The estimate $(0.297, 0.068, 0.536)$ was obtained for the parameter $(\alpha_0,\beta_0,\pi_0)$ with $(0.021, 0.002, 0.009)$ as the vector of estimated standard errors. The estimated regression line is represented by a solid line in the upper right plot of Figure~\ref{chipmix} while the dashed line represents the corresponding (transformed) regression line estimated in~\cite{MarMarBer08} under the assumption of normal errors. Note that the estimate of $\pi_0$ obtained therein is 0.32. The regression line obtained in~\cite[Section 4.4]{Van13} from a subsample of $n=30,000$ observations and under zero-symmetry assumptions for the errors is represented as a dotted line. The estimate $(F_n \vee 0) \wedge 1$ (resp.\ $f_n \vee 0$) of the unknown c.d.f.\ $F$ (resp.\ p.d.f.\ $f$) of $\varepsilon$ is represented in the lower left (resp.\ right) plot of Figure~\ref{chipmix} as a solid line. The computed approximate confidence band of level 0.95 for $F$ is not displayed because it cannot be distinguished from the estimated c.d.f.\ (which could have been expected given the huge sample size). The dashed curve in the lower left (resp.\ right) plot represents the c.d.f.\ (resp.\ p.d.f.) of the parametric estimation of the corresponding error, which is normal with standard deviation equal to $0.8$ as reported in~\cite[Section 4.4]{Van13}. The latter parametric c.d.f.\ lies clearly outside the confidence band. The estimation of $(\alpha_0,\beta_0,\pi_0,f,F)$, implemented in R, took less than 30 seconds on one 2.4 GHz processor while more than two days of computation on a similar processor were necessary in~\cite{Van13} to estimate the same parameters from a subsample of size $n=30,000$. Based on Figure~\ref{chipmix} and given the huge sample size, it seems sensible to reject both the assumptions of normality considered in~\cite{MarMarBer08} and the assumption of symmetry on which the method in~\cite{Van13} is based.

\section{Conclusion and possible extensions of the model}
\label{extension}

The identifiability of the model stated in~\eqref{model} was investigated and estimators of the Euclidean and functional parameters were proposed. The asymptotics of the latter were studied under weak conditions not involving zero-symmetry assumptions for the errors. In addition, a consistent procedure for computing an approximate confidence band for the c.d.f.\ of the error was proposed using a weighted bootstrap.

As mentioned by a referee, the model considered in this work is very specific. It is the constraint that the first component is assumed to be entirely known that enabled us to propose a relatively simple and numerically efficient estimation procedure. It is that same constraint that made it possible to obtain, unlike for EM-type algorithms, asymptotic results allowing to quantify the estimation uncertainty. The latter advantages clearly come at the price of a restricted applicability. As we shall see in the next subsection, it is possible in principle to improve this situation by introducing an unknown scale parameter for the first component. In the second subsection, we briefly discuss another extension of the model adapted to the situation where there is more than one explanatory variable.

\subsection{An additional unknown scale parameter for the first component}

From the illustrations presented in the previous section, we see that the price to pay for no parametric constraints on the second component is a complete specification of the first component. As mentioned in Section~\ref{PN}, from a theoretical perspective, it is possible to improve this situation by introducing an unknown scale parameter for the first component. Using the notation of Sections~\ref{PN} and~\ref{ident}, the extended model that we have in mind can be written as
\begin{equation}
\label{model2}
 Y = \left\{
\begin{array}{lll}
\sigma_0^* \bar \varepsilon^* &\mbox{ if } & Z=0, \\
\alpha_0+\beta_0 X+\varepsilon &\mbox{ if } & Z=1,
\end{array}\right. 
\end{equation} 
where $\bar \varepsilon^*$ is assumed to have variance one and known c.d.f.\ $\bar F$ while $\sigma_0^*$ is unknown. With respect to the model given in~(\ref{model}), this simply amounts to writing $\varepsilon^*$ as $\sigma_0^* \bar \varepsilon^*$ and the c.d.f.\ $F^*$ of $\varepsilon^*$ as $F^* = \bar F(\cdot/\sigma_0^*)$. The Euclidean parameter vector of this extended model is therefore $(\alpha_0,\beta_0,\pi_0,\sigma_0^*)$ and the functional parameter is $F$, the c.d.f.\ of $\varepsilon$.

The model given in~(\ref{model2}) is identifiable provided $\XX$, the set of possible values of $X$, contains four points ${x_1,x_2,x_3,x_4}$ such that the vectors $\{(1,x_i,x^2_i,x_i^3)\}_{1 \leq i \leq 4}$ are linearly independent. This can be verified by using, in addition to~(\ref{M1}) and~(\ref{M2}), the fact that
\begin{equation}
\label{M3}
\E(Y^3|X) = \pi_0\alpha_0(\alpha_0^2+3\sigma_0^2) +3\pi_0\beta_0(\alpha_0^2+\sigma_0^2) X +3\pi_0\alpha_0\beta_0^2X^2+\pi_0\beta_0^3X^3\quad \mbox{a.s.} 
\end{equation}
By proceeding as in Section~\ref{ident}, one can for instance show that 
\begin{equation}
\label{sigma*}
(\sigma_0^*)^2 = \frac{\lambda_{0,3}\lambda_{0,5}-\lambda_{0,7}\lambda_{0,2}}{\lambda_{0,5}-\lambda_{0,2}^2}, 
\end{equation}
where $\lambda_{0,2}$ is the coefficient of $X$ in~(\ref{M1}), $\lambda_{0,3}$ and $\lambda_{0,5}$ are the coefficients of $1$ and $X^2$, respectively, in~(\ref{M2}), and $\lambda_{0,7}$ is the coefficient of $X^2$ in~(\ref{M3}).

From a practical perspective however, using relationship~(\ref{sigma*}) for estimation (or a similar equation resulting from~(\ref{M1}),~(\ref{M2}) and~(\ref{M3})) turned out to be highly unstable. The reason why estimation of $\sigma_0^*$ by the method of moments does not work satisfactorily seems to be due to the fact that $(\sigma_0^*)^2$ is always the difference of two positive quantities. The estimation of each quantity is not precise enough to ensure that their difference is close to $(\sigma_0^*)^2$, and the difference is often negative. As an alternative estimation method, an iterative EM-type algorithm could be used to estimate all the unknown parameters of the extended model. Unfortunately, a weakness of such algorithms is that, up to now, the asymptotics of the resulting estimators are not known.

\subsection{More than one explanatory variable}

Assuming that there are $d$ explanatory random variables $X_1,\dots, X_d$ for some integer $d \geq 1$, and using the notation of Sections~\ref{PN} and~\ref{ident}, an immediate extension of the model stated in~\eqref{model} is
\begin{equation}
 Y = \left\{
\begin{array}{lll}
\varepsilon^* &\mbox{ if } & Z=0, \\
\vec \beta^\top \vec X+\varepsilon &\mbox{ if } & Z=1,
\end{array}\right. 
\end{equation}
where $\vec \beta = (\beta_0,\dots,\beta_d) \in \R^{d+1}$ is the Euclidean parameter of the unknown component and $\vec X = (1, X_1,\dots, X_d)$. Then, with the convention that $X_0 = 1$, we have
$$
\E(Y|X) = \sum_{i=0}^d \pi_0 \beta_i X_i \qquad \mbox{a.s.}
$$
and
$$
\E(Y^2|X) = (1- \pi_0) ( \sigma_0^* )^2 + \pi_0 (\beta_0^2 + \sigma_0^2) + \sum_{0 \leq i < j \leq d} 2 \pi_0 \beta_i \beta_j X_i X_j + \sum_{i=1}^d \pi_0 \beta_i^2 X_i^2 \quad \mbox{a.s.}
$$
Now, let
\begin{equation}
\label{system_extension}
\left\{
\begin{array}{ll}
\varrho_i = \pi_0 \beta_i, & i \in \{0,\dots,d\}, \\
\varsigma = (1- \pi_0) ( \sigma_0^* )^2 + \pi_0 (\beta_0^2 + \sigma_0^2)  & \\
\mu_{\{i,j\}} = 2 \pi_0 \beta_i \beta_j, & \{i,j\} \subset \{0,\dots,d\}, \\
\nu_i = \pi_0 \beta_i^2, & i \in \{1,\dots,d\}.
\end{array}
\right.
\end{equation}
Adapting {\em mutatis mutandis} the approach described in Section~\ref{ident}, we have that $\varrho_i$, $i \in \{0,\dots,d\}$, can be identified provided that the set $\XX \subset \R^d$ of possible values of $(X_1,\dots,X_d)$ is such that the space spaned by $\{(1,\vec x) : \vec x \in \XX\}$ is of dimension $d+1$. A similar (but painful to write) sufficient condition on $\XX$ can be stated ensuring that $\varsigma$, $\mu_{\{i,j\}}$, $\{i,j\} \subset \{0,\dots,d\}$, and $\nu_i$, $i \in \{1,\dots,d\}$, can be identified. Then, it can be verified that the system in~\eqref{system_extension} can be solved provided $\pi_0 \in (0,1]$ and there exists $k \in \{1,\dots,d\}$ such that $\beta_k \neq 0$. In that case, we obtain $\beta_k = \nu_k / \varrho_k$ and $\beta_j = \mu_{\{j,k\}} / 2 \varrho_k$ for any $j \in \{0,\dots,d\}$, $j \neq k$. In other words, a necessary condition to be able to identify $\vec \beta$ is that $\pi_0 \in (0,1]$ and there exists $k \in \{1,\dots,d\}$ such that $\beta_k \neq 0$. 

As far as estimation of $\vec \beta$ and $\pi_0$ is concerned, the system in~\eqref{system_extension} suggests a large number of possible estimators. Many additional estimators could be obtained by generalizing the approach used in the second half of Section~\ref{estimation}.

\section*{Acknowledgments}

The authors would like to thank two anonymous referees for their very insightful and constructive comments which helped to improve the paper.


\appendix

\section{Proof of Proposition~\ref{euclidean}}
\label{proof_euclidean}

\begin{proof}
Let us prove~(i). From Assumption A1~(i) and~(\ref{Y}), we have that $\E(X^pY^q)$ is finite for all integers $p,q \in \{0,1,2\}$. It follows that all the components of the vector of expectations $\E\{\dot \varphi_{\vec\gamma_0}(X,Y) \} = P \dot \varphi_{\vec\gamma_0}$ are finite. The strong law of large numbers then implies that $\P_n \dot \varphi_{\vec\gamma_0} \as P \dot \varphi_{\vec\gamma_0}$. Using the fact that $\vec \gamma_0$ is a zero of $\vec\gamma \mapsto P \dot\varphi_{\vec\gamma}$, that $\P_n \dot \varphi_{\vec\gamma_0} = \Gamma_n \vec\gamma_0 - \vec \theta_n$, and that $\P_n \dot \varphi_{\vec\gamma_n} = \Gamma_n \vec\gamma_n - \vec \theta_n = 0$, we obtain that $\Gamma_n(\vec\gamma_n-\vec\gamma_0) \as 0$. The strong law of large numbers also implies that $\Gamma_n \as \Gamma_0$. Matrix inversion being continuous with respect to any usual topology on the space of square matrices, Assumption A2 implies that $\Gamma_n^{-1} \as \Gamma_0^{-1}$. The continuous mapping theorem then implies that $\Gamma_n^{-1} \Gamma_n(\vec\gamma_n-\vec\gamma_0) = \vec\gamma_n-\vec\gamma_0 \as 0$. Since $\vec \gamma_0 \in \DD^{\alpha,\beta,\pi}$, the strong consistency of $(\alpha_n,\beta_n,\pi_n)$ is finally again a consequence of the continuous mapping theorem as the function  
\begin{equation}
\label{phi}
\vec \gamma \mapsto \left( g^\alpha, g^\beta, g^\pi \right) (\vec \gamma)  = (\alpha,\beta,\pi)
\end{equation}
from $\R^8$ to $\R^3$ is continuous on $\DD^{\alpha,\beta,\pi}$.

Let us now prove~(ii). Using the fact that $P \dot \varphi_{\vec\gamma_0} =0$ and $\P_n \dot \varphi_{\vec\gamma_n} = 0$, we have 
$$
\P_n \dot \varphi_{\vec\gamma_0} - P \dot \varphi_{\vec\gamma_0} =- (\P_n \dot \varphi_{\vec\gamma_n} - \P_n \dot \varphi_{\vec\gamma_0}) =- \P_n (\dot \varphi_{\vec\gamma_n} - \dot \varphi_{\vec\gamma_0})  = - \Gamma_n (\vec \gamma_n - \vec \gamma_0),
$$
which implies that $\G_n \dot \varphi_{\vec\gamma_0} = -\Gamma_n \sqrt{n}(\vec\gamma_n-\vec\gamma_0)$. From Assumption A1~(ii) and~(\ref{Y}), we have that the covariance matrix of the random vector $\dot \varphi_{\vec\gamma_0}(X,Y)$ is finite. The multivariate central limit theorem then implies that $\G_n \dot \varphi_{\vec\gamma_0}$ converges in distribution to a centered multivariate normal random vector $\G \dot \varphi_{\vec\gamma_0}$ with covariance matrix $P \dot \varphi_{\vec\gamma_0} \dot \varphi_{\vec\gamma_0}^\top$. Since $(\G_n \dot \varphi_{\vec\gamma_0},\Gamma_n) \leadsto (\G \dot \varphi_{\vec\gamma_0},\Gamma_0)$ and under Assumption~A2, we obtain, from the continuous mapping theorem, that  
$$
\sqrt{n}(\vec\gamma_n-\vec\gamma_0) = -\Gamma_n^{-1} \G_n \dot \varphi_{\vec\gamma_0} \leadsto -\Gamma_0^{-1} \G \dot \varphi_{\vec\gamma_0}.
$$
The map defined in~(\ref{phi}) is differentiable at $\vec \gamma_0$ since $\vec \gamma_0 \in \DD^{\alpha,\beta,\pi}$. We can thus apply the delta method with that map to obtain that
$$
\sqrt{n}(\alpha_n-\alpha_0,\beta_n-\beta_0,\pi_n-\pi_0) = - \Psi_{\vec \gamma_0} \Gamma_n^{-1} \G_n \dot \varphi_{\vec\gamma_0} + o_P(1).
$$
Since $\Gamma_n^{-1} \as \Gamma_0^{-1}$ under Assumption A2, we obtain that 
$$
\sqrt{n}(\alpha_n-\alpha_0,\beta_n-\beta_0,\pi_n-\pi_0) = - \Psi_{\vec \gamma_0} \Gamma_0^{-1} \G_n \dot \varphi_{\vec\gamma_0} + o_P(1).
$$
It remains to prove that $\Sigma_n \as \Sigma$. Under Assumption A1~(ii), the strong law of large numbers implies that $\P_n \dot \varphi_{\vec\gamma_0} \dot \varphi_{\vec\gamma_0}^\top \as P \dot \varphi_{\vec\gamma_0} \dot \varphi_{\vec\gamma_0}^\top$. The fact that $\P_n \dot \varphi_{\vec\gamma_n} \dot \varphi_{\vec\gamma_n}^\top = \P_n \dot \varphi_{\vec\gamma_0} \dot \varphi_{\vec\gamma_0}^\top  + \P_n (\dot \varphi_{\vec\gamma_n} \dot \varphi_{\vec\gamma_n}^\top - \dot \varphi_{\vec\gamma_0} \dot \varphi_{\vec\gamma_0}^\top ) \as P \dot \varphi_{\vec\gamma_0} \dot \varphi_{\vec\gamma_0}^\top$ is then a consequence of the fact that $\vec \gamma_n \as \vec \gamma_0$ and the continuous mapping theorem. Similarly, since $\vec \gamma_0 \in \DD^{\alpha,\beta,\gamma}$, we additionally have that $\Psi_{\vec \gamma_n} \as \Psi_{\vec \gamma_0}$. Combined with the fact that, under Assumption A2, $\Gamma_n^{-1} \as \Gamma_0^{-1}$, we obtain that  $\Sigma_n \as \Sigma$ from the continuous mapping theorem.
\end{proof}


\section{Proof of Proposition~\ref{functional}}
\label{proof_functional}

The proof of Proposition~\ref{functional} is based on three lemmas.

\begin{lem}
\label{lemDonsker}
The classes of functions $\FF^J$ and $\FF^K$ are $P$-Donsker. So is the class $\FF^{\alpha,\beta,\pi}$ provided Assumptions A1~(ii) and A2 hold, and $\vec \gamma_0 \in \DD^{\alpha,\beta,\pi}$. 
\end{lem}

\begin{proof}
The class $\FF^J$ is the  class of indicator functions $(x,y) \mapsto \1\{(x,y) \in C_{t,\vec \eta}\}$, where $C_{t,\vec \eta}= \{(x,y) \in \R^2 : y \leq t + \alpha + \beta x\}$. The collection $\CC =\{C_{t,\vec \eta} : t \in \R, \vec \eta = (\alpha,\beta) \in \R^2 \}$ is the set of all half-spaces in $\R^2$. From~\cite[Exercise~14, p 152]{vanWel96}, it is a Vapnik-\v Chervonenkis ($VC$) class with $VC$ dimension 4. By Lemma~9.8 in~\cite{Kos08}, $\FF^J$ has the same $VC$ dimension as $\CC$. Being a set of indicator functions, $\FF^J$ clearly possesses a square integrable envelope function and is therefore $P$-Donsker.

The class $\FF^K$ is a collection of monotone functions, and it is easy to verify that it has $VC$ dimension 1. Furthermore, it clearly possesses a square integrable envelope function because the elements of $\FF^K$ are bounded. It is therefore $P$-Donsker. 

The components classes of class $\FF^{\alpha,\beta,\pi}$ are well defined since Assumption A2 holds and $\vec \gamma_0 \in \DD^{\alpha,\beta,\pi}$. It is easy to see that they are linear combinations of a finite collection of functions that, from Assumption A1~(ii), is $P$-Donsker. The components classes of $\FF^{\alpha,\beta,\pi}$ are therefore $VC$ classes. They possess square integrable envelope functions because $\DD^{\alpha,\beta,\pi}_{\vec \gamma_0}$ is a bounded set. The class $\FF^{\alpha,\beta,\pi}$ is therefore $P$-Donsker.
\end{proof}


\begin{lem} 
\label{lemQuad}
Under Assumptions A1~(i) and A3~(i), 
$$
\sup_{t\in\R} P(\psi_{t,\vec \eta}^J-\psi_{t, \vec \eta_0}^J)^2 \to 0 \qquad \mbox{and} \qquad 
\sup_{t\in\R} P(\psi_{t,\vec \eta}^K-\psi_{t, \vec \eta_0}^K)^2 \to 0 \qquad \mbox{as} \qquad \vec \eta \to \vec \eta_0.
$$
\end{lem}

\begin{proof}
For class $\FF^J$, for any $t\in\R$, we have
{\small \begin{align*}
P&(\psi^J_{t,\vec \eta}-\psi^J_{t,\vec \eta_0})^2 = | P(\psi^J_{t,\vec \eta}+\psi^J_{t,\vec \eta_0}-2\psi^J_{t,\vec \eta}\psi^J_{t,\vec \eta_0}) | \\
=&  P\{(\psi^J_{t,\vec \eta}-\psi^J_{t,\vec \eta_0}) \1(\alpha_0+\beta_0 x<\alpha+\beta x)\} + P\{(\psi^J_{t,\vec \eta_0}-\psi^J_{t,\vec \eta}) \1(\alpha_0+\beta_0 x>\alpha+\beta x)\} \\
=&  \int_{\R}\left\{F_{Y|X}(t+\alpha+\beta x|x)-F_{Y|X}(t+\alpha_0+\beta_0 x|x)\right\} \1(\alpha_0+\beta_0 x<\alpha+\beta x) \dd F_X(x) \\
&+  \int_{\R}\left\{F_{Y|X}(t+\alpha_0+\beta_0 x|x)-F_{Y|X}(t+\alpha+\beta x|x)\right\} \1(\alpha_0+\beta_0 x>\alpha+\beta x) \dd F_X(x) \\
\leq& \int_{\R}\left | F_{Y|X}(t+\alpha_0+\beta_0 x|x)-F_{Y|X}(t+\alpha+\beta x|x)\right | \dd F_X(x),
\end{align*}}
where $F_{Y|X}$ is defined in~(\ref{FYX}). Since $f_{Y|X}(\cdot|x)$ defined in~(\ref{fYX}) exists for all $x\in \XX$, the mean value theorem enables us to write, for any $t \in \R$ and $x\in \XX$, 
\begin{multline*}
F_{Y|X}(t + \alpha + \beta x | x)-F_{Y|X}(t + \alpha_0 + \beta_0 x | x) = f_{Y|X}(t + \tilde \alpha_{x,t} + \tilde \beta_{x,t} x | x) \\ \times \left\{ (\alpha - \alpha_0) + x (\beta - \beta_0) \right\},
\end{multline*}
where $\tilde \alpha_{x,t} + \tilde \beta_{x,t} x$ is between $\alpha + \beta x$ and $\alpha_0 + \beta_0 x$. It follows that
\begin{align*}
\sup_{t\in\R} P(\psi_{t,\vec \eta}^J&-\psi_{t, \vec \eta_0}^J)^2 \\ 
&\leq \sup_{t\in\R}\int_{\R} f_{Y|X}(t + \tilde \alpha_{x,t} + \tilde \beta_{x,t} x | x) \left| (\alpha - \alpha_0) + x (\beta - \beta_0) \right| \dd F_X(x) \\ 
&\leq \left\{ \sup_{t\in\R} f^*(t) + \sup_{t\in\R} f(t) \right\}\left\{ | \alpha - \alpha_0 | + \E(|X|) |\beta - \beta_0| \right\}.
\end{align*}
Under Assumption A3~(i), the supremum on the right of the previous display is finite and, under Assumption A1~(i), so is $\E(|X|)$. We therefore obtain the desired result.

For class $\FF^K$, we have 
\begin{align*}
\sup_{t\in\R} P(\psi_{t,\vec \eta}^K-\psi_{t, \vec \eta_0}^K)^2 &=  \int_\R \{F^*(t+ \alpha+\beta x) -F^*(t+ \alpha_0+\beta_0 x) \}^2 \dd F_X(x) \\
&\leq  \int_\R | F^*(t+ \alpha+\beta x) -F^*(t+ \alpha_0+\beta_0 x) | \dd F_X(x),
\end{align*}
from the convexity of $x \mapsto x^2$ on $[0,1]$. Proceeding as previously, by the mean value theorem, we obtain that
$$
\sup_{t\in\R} P(\psi_{t,\vec \eta}^K-\psi_{t, \vec \eta_0}^K)^2 \leq \left\{ \sup_{t\in\R} f^*(t) \right\}\left\{ | \alpha - \alpha_0 | + \E(|X|) |\beta - \beta_0| \right\}.
$$
Under Assumptions A1~(i) and A3~(i), the right-hand side of the previous inequality tends to zero as $\vec \eta \to \vec \eta_0$.
\end{proof}


\begin{lem}
\label{lemJK}
Under Assumptions A1~(ii), A2 and A3~(ii), for any $t \in \R$,
\begin{multline*}
\sqrt{n} \{J_n(\vec \eta_n,t) - J(\vec \eta_0,t)\} = \sqrt{n} \left( \P_n \psi_{t,\vec \eta_n}^J - P \psi_{t,\vec \eta_0}^J \right) \\ = \G_n \left( \psi_{t,\vec \eta_0}^J + \left[ (1-\pi_0) \E \{ f^*(t+\alpha_0 + \beta_0 X) \} + \pi_0 f(t) \right] \psi_{\vec \gamma_0}^\alpha \right. \\ \left. + \left[ (1 - \pi_0) \E \{ X f^*(t+\alpha_0 + \beta_0 X) \} + \pi_0 f(t) \E (X) \right] \psi_{\vec \gamma_0}^\beta \right) + R_{n,t}^J, 
\end{multline*}
and
\begin{multline*}
\sqrt{n} \{K_n(\vec \eta_n,t) - K(\vec \eta_0,t)\} = \sqrt{n} \left( \P_n \psi_{t,\vec \eta_n}^K - P \psi_{t,\vec \eta_0}^K \right) \\ = \G_n \left( \psi_{t,\vec \eta_0}^K + \E\{f^*(t+\alpha_0 + \beta_0 X)\} \psi_{\vec \gamma_0}^\alpha  + \E\{X f^*(t+\alpha_0 + \beta_0 X)\} \psi_{\vec \gamma_0}^\beta \right) + R_{n,t}^K,
\end{multline*}
where $\sup_{t \in \R} |R_{n,t}^J| \p 0$ and $\sup_{t \in \R} |R_{n,t}^K| \p 0$.
\end{lem}

\begin{proof}
We only prove the first statement as the proof of the second statement is similar. For any $t \in \R$, we have
$$
\sqrt{n} \left( \P_n \psi_{t,\vec \eta_n}^J - P \psi_{t,\vec \eta_0}^J \right) = \G_n \left( \psi_{t,\vec \eta_n}^J - \psi_{t,\vec \eta_0}^J \right) + \G_n \psi_{t,\vec \eta_0}^J + \sqrt{n}P\left( \psi_{t,\vec \eta_n}^J - \psi_{t,\vec \eta_0}^J \right).
$$
Using the fact that $\vec \eta_n \as \vec \eta_0$, Lemma~\ref{lemDonsker}, and Lemma~\ref{lemQuad}, we can apply Theorem 2.1 in~\cite{vanWel07} to obtain that 
$$
\sup_{t \in \R} \left| \G_n \left( \psi_{t,\vec \eta_n}^J - \psi_{t,\vec \eta_0}^J \right) \right| \p 0.
$$
Furthermore, for any $t \in \R$, we have 
\begin{multline*}
\sqrt{n}P\left( \psi_{t,\vec \eta_n}^J - \psi_{t,\vec \eta_0}^J \right) \\ = \sqrt{n} \int_\R \left\{ F_{Y|X}(t+\alpha_n + \beta_n x | x) -  F_{Y|X}(t+\alpha_0 + \beta_0 x | x) \right\} \dd F_X(x),
\end{multline*}
where $F_{Y|X}$ is defined in~(\ref{FYX}). Since $f_{Y|X}'(\cdot|x)$, the derivative of $f_{Y|X}(\cdot|x)$,  exists for all $x \in \XX$ from Assumption A3~(ii) and~(\ref{fYX}), we can apply the second-order mean value theorem to obtain
\begin{multline*}
\sqrt{n}P\left( \psi_{t,\vec \eta_n}^J - \psi_{t,\vec \eta_0}^J \right) \\ =  \sqrt{n} \int_\R f_{Y|X}(t+\alpha_0 + \beta_0 x | x) \{(\alpha_n-\alpha_0) +(\beta_n - \beta_0) x \} \dd F_X(x) + R^J_{n,t},
\end{multline*}
where 
$$
R^J_{n,t} = \frac{\sqrt{n}}{2} \int_\R f_{Y|X}'(t+ \tilde \alpha_{x,t,n} + \tilde \beta_{x,t,n} x | x)  \{ (\alpha_n-\alpha_0) + (\beta_n - \beta_0) x \}^2 \dd F_X(x),
$$
and $\tilde \alpha_{x,t,n} + \tilde \beta_{x,t,n} x$ is between $\alpha_0 + \beta_0 x$ and $\alpha_n + \beta_n x$. Now, from~(\ref{fYX}),
\begin{multline*}
\sup_{t \in \R} |R^J_{n,t}| \leq \sqrt{n} \left\{ \sup_{t \in \R} (f^*)'(t) + \sup_{t \in \R} f'(t) \right\} \\ \times \left\{ (\alpha_n-\alpha_0)^2 + (\beta_n - \beta_0)^2 \E(X^2) + 2 |\alpha_n - \alpha_0| |\beta_n - \beta_0| \E(|X|) \right\}. 
\end{multline*}
The supremum on the right of the previous inequality is finite from Assumption A3~(ii), and so are $\E(|X|)$ and $\E(X^2)$ from Assumption A1~(ii). Furthermore, under Assumptions A1~(ii) and A2, we know from Proposition~\ref{euclidean} that $\sqrt{n} (\alpha_n-\alpha_0, \beta_n-\beta_0)$ converges in distribution while $(\alpha_n, \beta_n) \as (\alpha_0,\beta_0)$. It follows that $\sup_{t \in \R} |R^J_{n,t}|\p 0$. Hence, we obtain that
\begin{multline*}
\sqrt{n}P\left( \psi_{t,\vec \eta_n}^J - \psi_{t,\vec \eta_0}^J \right) = \E\{f_{Y|X}(t+\alpha_0 + \beta_0 X | X)\} \sqrt{n} (\alpha_n - \alpha_0) \\ + \E\{X f_{Y|X}(t+\alpha_0 + \beta_0 X | X)\} \sqrt{n} (\beta_n - \beta_0) + R^J_{n,t}, \qquad t \in \R.
\end{multline*}
The desired result finally follows from the expression of $f_{Y|X}$ given in~(\ref{fYX}) and Proposition~\ref{euclidean}.
\end{proof}

\begin{proof}[\bf Proof of Proposition~\ref{functional}]
Under Assumptions  A1~(ii) and A2, and since $\vec \gamma_0 \in \DD^{\alpha,\beta,\pi}$, we know, from Lemma~\ref{lemDonsker}, that the classes $\FF^J$, $\FF^K$ and $\FF^{\alpha,\beta,\pi}$ are $P$-Donsker. It follows that
$$
 \left( t \mapsto \G_n \psi_{t,\vec \eta_0}^J, t \mapsto \G_n \psi_{t,\vec \eta_0}^K, \G_n \psi_{\vec \gamma_0}^\alpha,  \G_n \psi_{\vec \gamma_0}^\beta, \G_n \psi_{\vec \gamma_0}^\pi  \right)
$$
converges weakly in $\{\ell^\infty(\overline{\R})\}^2 \times \R^3$. Assumption A3~(i) then implies that the functions $t \mapsto \E\{f_{Y|X}(t+\alpha_0 + \beta_0 X | X)\}$, $t \mapsto \E\{X f_{Y|X}(t+\alpha_0 + \beta_0 X | X)\}$, $t \mapsto \E\{f^*(t+\alpha_0 + \beta_0 X)\}$, and $t \mapsto \E\{X f^*(t+\alpha_0 + \beta_0 X)\}$ are bounded. By the continuous mapping theorem, we thus obtain that
{\small$$
\left( 
\begin{array}{c} t \mapsto \G_n \left( \psi_{t,\vec \eta_0}^J + \E\{f_{Y|X}(t+\alpha_0 + \beta_0 X | X)\} \psi_{\vec \gamma_0}^\alpha + \E\{X f_{Y|X}(t+\alpha_0 + \beta_0 X | X) \} \psi_{\vec \gamma_0}^\beta \right) \\ 
t \mapsto \G_n \left( \psi_{t,\vec \eta_0}^K + \E\{f^*(t+\alpha_0 + \beta_0 X)\} \psi_{\vec \gamma_0}^\alpha + \E\{X f^*(t+\alpha_0 + \beta_0 X)\} \psi_{\vec \gamma_0}^\beta \right) \\ 
\G_n \psi_{\vec \gamma_0}^\pi 
\end{array} \right)
$$}
converges weakly in $\{ \ell^\infty(\overline{\R}) \}^2 \times \R$. It follows from Proposition~\ref{euclidean} and Lemma~\ref{lemJK} that
$$
\sqrt{n} \left( J_n(\vec \eta_n, \cdot) - J(\vec \eta_0, \cdot), K_n(\vec \eta_n, \cdot) - K(\vec \eta_0, \cdot),\pi_n - \pi_0 \right),
$$
converges weakly in $\{ \ell^\infty(\overline{\R}) \}^2 \times \R$. The desired result is finally a consequence of~(\ref{Fn}) and the functional delta method applied with the map $(J,K,\pi) \mapsto \left\{ J  - (1-\pi) K \right\} / \pi$.
\end{proof}


\section{Proof of Proposition~\ref{pdf}}
\label{proof_pdf}
\begin{proof}
The assumptions of Proposition~\ref{euclidean} being verified, we have that $\pi_n \as \pi_0 \neq 0$. Then, as can be verified from~(\ref{fn}), to show the desired result, it suffices to show that
\begin{multline*}
\sup_{t \in \R} \left| \frac{1}{nh_n} \sum_{i=1}^n \kappa \left(\frac{t-Y_i+\alpha_n+\beta_n X_i}{h_n}\right) \right. \\ \left. -\frac{(1-\pi_0)}{n} \sum_{i=1}^n f^*(t+\alpha_n+\beta_n X_i) - \pi_0 f(t) \right| \as 0.
\end{multline*}
The previous supremum is smaller than $I_n + (1 - \pi_0) II_n$, where
\begin{multline*}
I_n = \sup_{t \in \R} \left| \frac{1}{nh_n} \sum_{i=1}^n \kappa \left(\frac{t-Y_i+\alpha_n+\beta_n X_i}{h_n}\right) \right. \\ \left. - (1- \pi_0)\int_\R f^*(t + \alpha_0 + \beta_0 x) f_X(x) \dd x - \pi_0 f(t) \right|,
\end{multline*}
and 
$$
II_n = \sup_{t \in \R} \left| \frac{1}{n} \sum_{i=1}^n f^*(t+\alpha_n+\beta_n X_i) -  \int_\R f^*(t + \alpha_0 + \beta_0 x) f_X(x) \dd x \right|. 
$$
Let us first show that $I_n \as 0$. Consider the class $\FF$ of measurable functions from $\R^2$ to $\R$ defined by
\begin{multline*}
\FF = \left\{ (x,y) \mapsto \psi_{\vec \eta,t,h}(x) = \kappa \left( \frac{t-y+\alpha+\beta x}{h} \right) : \vec \eta = (\alpha,\beta) \in \R^2, \right. \\ \left. t \in \R, h \in (0, \infty)\right\},
\end{multline*}
and notice that
$$
\P_n \psi_{\vec \eta_n,t,h_n} = \frac{1}{n} \sum_{i=1}^n \kappa \left(\frac{t-Y_i+\alpha_n+\beta_n X_i}{h_n}\right), \qquad t \in \R,
$$
where $\vec \eta_n = (\alpha_n,\beta_n)$. Then, $I_n \leq I_n' + I_n''$, where
\begin{equation}
\label{In'}
I_n' = \frac{1}{h_n} \sup_{t \in \R} \left|  \P_n \psi_{\vec \eta_n,t,h_n} - P \psi_{\vec \eta_n,t,h_n} \right| = \frac{1}{h_n \sqrt{n}} \sup_{t \in \R} \left|  \G_n \psi_{\vec \eta_n,t,h_n} \right|, 
\end{equation}
and
$$
I_n'' = \sup_{t \in \R} \left| \frac{1}{h_n} P \psi_{\vec \eta_n,t,h_n} - g(t) \right|,
$$
with
$$
g(t) = (1- \pi_0)\int_\R f^*(t + \alpha_0 + \beta_0 x) f_X(x) \dd x + \pi_0 f(t), \qquad t \in \R.
$$
Let us first deal with $I_n''$.  From~(\ref{fYX}), notice that 
$$
g(t) = \int_\R f_{Y|X}(t + \alpha_0 + \beta_0 x | x) f_X(x) \dd x, \qquad t \in \R.
$$ 
Also, for any $t \in \R$,
$$
P \psi_{\vec \eta_n,t,h_n} = \int_\R  \left\{ \int_\R \kappa \left( \frac{t-y+\alpha_n+\beta_n x}{h_n} \right) f_{Y | X} (y | x) \dd y \right\} f_X(x) \dd x ,
$$
which, using the change of variable $u = (t-y+\alpha_n+\beta_n x)/h_n$ in the inner integral, can be rewritten as
$$
P \psi_{\vec \eta_n,t,h_n} = h_n \int_\R  \left\{ \int_\R  \kappa (u) f_{Y | X} (t + \alpha_n + \beta_n x - u h_n | x) \dd u \right\} f_X(x) \dd x.
$$
Since $\kappa$ is a p.d.f.\ from Assumption A4~(ii), it follows that, for any $t \in \R$, 
\begin{multline*}
\frac{1}{h_n} P \psi_{\vec \eta_n,t,h_n} - g(t) = \\ \int_\R  \left[ \int_\R \kappa (u) \left\{ f_{Y | X} (t + \alpha_n + \beta_n x - u h_n | x) - f_{Y|X}(t + \alpha_0 + \beta_0 x | x) \right\} \dd u \right] f_X(x) \dd x.
\end{multline*}
As $f_{Y|X}'(\cdot|x)$, the derivative of $f_{Y|X}(\cdot|x)$, exists for all $x \in \XX$ under Assumption A3~(ii), the mean value theorem enables us to write
\begin{multline*}
I_n''  \leq \left\{ \sup_{t \in \R} (f^*)'(t) + \sup_{t \in \R} f'(t) \right\} \\ \times \int_\R  \left[ \int_\R \kappa (u) \left\{ |\alpha_n - \alpha_0| + |\beta_n - \beta_0| |x| + |u| h_n  \right\} \dd u \right] f_X(x) \dd x.
\end{multline*}
Hence,
\begin{multline*}
I_n''  \leq \left\{ \sup_{t \in \R} (f^*)'(t) + \sup_{t \in \R} f'(t) \right\}  \\ \times \left\{ |\alpha_n - \alpha_0| + |\beta_n - \beta_0| \E(|X|) + h_n \int_\R |u| \kappa(u) \dd u  \right\},
\end{multline*}
which, from Assumptions A1~(i), A3~(ii), A4~(ii), and Proposition~\ref{euclidean}~(i), implies that $I_n'' \as 0$.

Let us now show that $I_n' \as 0$. Since $\kappa$ has bounded variations from Assumption A4~(ii), it can be written as $\kappa_1-\kappa_2$, where both $\kappa_1$ and $\kappa_2$ are bounded nondecreasing functions on~$\R$. Without loss of generality, we shall assume that $\kappa$, $\kappa_1$ and $\kappa_2$ are bounded by 1. Then, for $j=1,2$, we define
$$
\FF_j=\left\{ (x,y) \mapsto \kappa_j\left(\frac{t-y+\alpha+\beta x}{h}\right) : (\alpha,\beta,t) \in \R^3, h \in (0, \infty)\right\}.
$$ 
Proceeding as in~\cite[proof of Lemma~22]{NolPol87}, let us first show that $\FF_j$ is a $VC$ class for $j=1,2$. Let $\kappa^-_j$ be the generalized inverse of $\kappa_j$ defined by $\kappa^-_j(c)=\inf\{x\in\R : \kappa_j(x)\geq c\}$, $c \in \R$. We consider the partition $\{C_1,C_2\}$ of $\R$ defined by
$$
\{ x \in \R : \kappa_j(x)>c \}= 
\left\{
\begin{array}{lll}
(\kappa_j^-(c), \infty ) & \mbox{ if } & c \in C_1,\\
\left[ \kappa_j^-(c), \infty ) \right.& \mbox{ if } & c \in C_2.
\end{array}
\right.
$$
Given $(\alpha,\beta,t) \in \R^3$ and $h \in (0, \infty)$, the set
\begin{equation}
\label{set}
\left\{ (x,y,c) \in \R^3 : \kappa_j \left( \frac{t-y+\alpha+\beta x}{h} \right) > c \right\} 
\end{equation}
can therefore be written as the union of 
$$
\left\{ (x,y,c) \in \R^2 \times C_1 : t - y + \alpha + \beta x - h \kappa_j^-(c) > 0 \right\}
$$
and
$$
\left\{ (x,y,c) \in \R^2 \times C_2 : t - y + \alpha + \beta x - h\kappa_j^-(c) \geq 0 \right\}.
$$
Now, let $f_{\alpha,\beta,t,h}(x,y,c) = t - y + \alpha + \beta x - h\kappa_j^-(c)$. The functions $f_{\alpha,\beta,t,h}$, with $(\alpha,\beta,t) \in \R^3$ and $h \in (0, \infty)$, span a finite-dimensional vector space. Hence, from Lemma~18~(ii) in~\cite{NolPol87}, the collections of all sets $\{(x,y,c) \in \R^2 \times C_1 : f_{\alpha,\beta,t,h}(x,y,c) > 0\}$ and $\{(x,y,c) \in \R^2 \times C_2 : f_{\alpha,\beta,t,h}(x,y,c) \geq 0\}$ are $VC$ classes. It follows that the collection of subgraphs of $\FF_j$ defined by~(\ref{set}), and indexed by $(\alpha,\beta,t) \in \R^3$ and $h \in (0, \infty)$, is also $VC$, which implies that $\FF_j$ is a $VC$ class of functions.

Given a probability distribution $Q$ on $\R^2$, recall that $L_2(Q)$ is the norm defined by $(Q f^2)^{1/2}$, with $f$ a measurable function from $\R^2$ to $\R$. Given a class $\GG$ of measurable functions from $\R^2$ to $\R$, the {\em covering number} $N ( \varepsilon, \GG, L_2(Q) )$ is the minimal number of $L_2(Q)$-balls of radius $\varepsilon > 0$ needed to cover the set $\GG$. From Lemma~16 in~\cite{NolPol87}, since $\FF = \FF_1 - \FF_2$, and since $\FF_1$ and $\FF_2$ have for an envelope the constant function 1 on $\R^2$, we have 
$$
\sup_{Q}N( 2\varepsilon,\FF, L_2(Q))\leq \sup_{Q}N(\varepsilon,\FF_1, L_2(Q))\times \sup_{Q}N(\varepsilon,\FF_2, L_2(Q)),
$$
for  probability measures $Q$ on $\R^2$. Using the fact that both $\FF_1$ and $\FF_2$ are $VC$ classes of functions with constant envelope 1, from Theorem~2.6.7 in~\cite{vanWel96} (see also the discussion on the top of page 246), we obtain that there exist constants $u > 0$ and $v > 0$ that depend on $\FF_1$ and $\FF_2$ such that
$$
\sup_{Q}N(\varepsilon,\FF, L_2(Q))\leq \left(\frac{u}{\varepsilon}\right)^{v}, \qquad\mbox{for every } 0<\varepsilon<u.
$$  
Then, by Theorem~2.14.9 in~\cite{vanWel96}, there exists constants $c_1 > 0$ and $c_2 > 0$ such that, for $\varepsilon>0$ large enough, 
$$
\Pr^*\left( \sup_{f \in \FF} |\G_n f| > \varepsilon \right) \leq c_1 \varepsilon^{c_2}\exp(-2\varepsilon^2).
$$
Starting from~(\ref{In'}), we thus obtain that, for every $\varepsilon>0$ and $n$ large enough,
\begin{multline*}
\Pr^*(I_n' > \varepsilon) = \Pr^* \left( \sup_{t \in \R} \left|  \G_n \psi_{\vec \eta_n,t,h_n} \right| >  \sqrt{n}h_n\varepsilon \right) \\ \leq \Pr^*\left( \sup_{f\in\FF} |\G_n f | >\sqrt{n}h_n\varepsilon\right) \leq c_1(\sqrt{n}h_n \varepsilon)^{c_2}\exp(-2nh_n^2\varepsilon^2)=a_n.
\end{multline*}
From Assumption A4~(i), it can be verified that $a_{n+1}/a_n \to 1$ and that $n(a_{n+1}/a_n-1)\to -\infty$. It follows from Raabe's rule that the series with general term $a_n$ converges. The Borel-Cantelli lemma enables us to conclude that $I_n' \as 0$, and we therefore obtain that $I_n \as 0$.

Since $f^*$ has bounded variations from Assumption A4~(iii), one can proceed along the same lines to show that $II_n \as 0$.
\end{proof}


\section{Proof of Proposition~\ref{mult}}
\label{proof_mult}

The proof of Proposition~\ref{mult} is based on the following lemma.

\begin{lem}
\label{extended}
Let $\Theta \subset \R^p$ and $H_0 \subset \R^q$ for some integers $p,q > 0$, let $\FF = \{f_{\theta,\zeta} : \theta \in \Theta\, , \zeta \in H_0 \}$ be a class of measurable functions from $\R^2$ to $\R$, and let $\zeta_n$ be an estimator of $\zeta_0 \in H_0$ such that $\Pr(\zeta_n \in H_0) \to 1$. If $\FF$ is $P$-Donsker and 
$$
\sup_{\theta \in \Theta} P(f_{\theta,\zeta_n} -f_{\theta,\zeta_0} )^2 \p 0 ,
$$
then,
$$
\sup_{\theta \in \Theta} \left| \G_n' (f_{\theta,\zeta_n} -f_{\theta,\zeta_0} ) \right| \p 0.
$$
\end{lem}

\begin{proof}
The result is the analogue of Theorem 2.1 in~\cite{vanWel07} in which $\G_n$ is replaced by $\G_n'$. The proof of Theorem 2.1 relies on the fact that $\G_n \leadsto \G$ in $\ell^\infty(\FF)$ and on the uniform continuity of the sample paths of the $P$-Brownian bridge $\G$; see~\cite[proof of Theorem 19.26]{van98} and~\cite{van02}. From the functional multiplier central limit theorem~\cite[see e.g.][Theorem 10.1]{Kos08}, we know that $(\G_n,\G_n')$ converges weakly in $\{ \ell^\infty(\FF) \}^2$ to $(\G,\G')$, where $\G'$ is an independent copy of the $\G$. The desired result therefore follows from a straightforward adaptation of the proof of Theorem 2.1 in~\cite{vanWel07}.
\end{proof}

\begin{proof}[\bf Proof of Proposition~\ref{mult}]
Since Assumptions A1~(ii) and A2 hold, we have from Lemma~\ref{lemDonsker} that $\FF^J$, $\FF^K$ and $\FF^{\alpha,\beta,\pi}$ are $P$-Donsker. Furthermore, $\E(X)$ is finite from Assumption A1~(i), the function $f$ is bounded from Assumption A3~(i), and so is the function $t \mapsto P (\psi_{t,\vec \eta_0}^K  - \psi_{t,\vec \eta_0}^J)$ from the definitions of $J$ and $K$ given in~(\ref{J}) and~(\ref{K}). Hence, from the functional multiplier central limit theorem~\cite[see e.g.][Theorem 10.1]{Kos08} and the continuous mapping theorem, we obtain that
$$
\left( t \mapsto \G_n  \psi_{t,\vec \gamma_0}^F, t \mapsto \G_n' \psi_{t,\gamma_0}^F \right) \leadsto \left( t \mapsto \G  \psi_{t,\vec \gamma_0}^F, t \mapsto \G' \psi_{t,\vec \gamma_0}^F \right)
$$
in $\{ \ell^\infty (\overline{\R}) \}^2$, where $\psi_{t,\vec \gamma_0}^F$ is defined in~(\ref{psiF}) and $t \mapsto \G' \psi_{t,\vec \gamma_0}^F$ is an independent copy of $t \mapsto \G  \psi_{t,\vec \gamma_0}^F$. It remains to show that 
$$
\sup_{t \in \R} \left| \G_n' \left( \hat \psi_{t,\vec \gamma_n}^F - \psi_{t,\vec \gamma_0}^F\right) \right| \p 0 . 
$$
From~(\ref{psiF}) and~(\ref{hatpsiF}), for any $t \in \R$, we can write
\begin{multline}
\label{diff}
\left| \G_n' \left( \hat \psi_{t,\vec \gamma_n}^F - \psi_{t,\vec \gamma_0}^F \right) \right| \leq  \left| \G_n' \left( \frac{1}{\pi_n} \psi_{t,\vec \eta_n}^J - \frac{1}{\pi_0} \psi_{t,\vec \eta_0}^J \right) \right| \\ + \left| \G_n' \left( f_n(t) \hat \psi_{\vec \gamma_n}^\alpha  - f(t) \psi_{\vec \gamma_0}^\alpha \right) \right| +  \left| \G_n' \left( f_n(t) \bar X \hat \psi_{\vec \gamma_n}^\beta - f(t) \E(X) \psi_{\vec \gamma_0}^\beta \right) \right| \\ + \left| \G_n' \left(  \frac{1 - \pi_n}{\pi_n}  \psi_{t,\vec \eta_n}^K - \frac{1 - \pi_0}{\pi_0}  \psi_{t,\vec \eta_0}^K \right) \right| \\ + \left| \G_n' \left( \frac{\P_n \psi_{t,\vec \eta_n}^K  - \P_n \psi_{t,\vec \eta_n}^J}{\pi_n^2} \hat \psi_{\vec \gamma_n}^\pi - \frac{P \psi_{t,\vec \eta_0}^K  - P \psi_{t,\vec \eta_0}^J}{\pi_0^2} \psi_{\vec \gamma_0}^\pi  \right) \right|.
\end{multline}
The last absolute value on the right of the previous display is smaller than
\begin{multline}
\label{last}
\left|  \frac{\P_n \psi_{t,\vec \eta_n}^K  - \P_n \psi_{t,\vec \eta_n}^J}{\pi_n^2} - \frac{P \psi_{t,\vec \eta_0}^K  - P \psi_{t,\vec \eta_0}^J}{\pi_0^2}  \right| \left|  \G_n' \psi_{\vec \gamma_0}^\pi \right| \\ + \left| \frac{P \psi_{t,\vec \eta_0}^K  - P \psi_{t,\vec \eta_0}^J}{\pi_0^2} \right| \left| \G_n' \left(  \hat \psi_{\vec \gamma_n}^\pi - \psi_{\vec \gamma_0}^\pi  \right) \right|. 
\end{multline}
Now,
\begin{multline}
\label{tmp}
\sup_{t \in \R} \left|  \P_n \psi_{t,\vec \eta_n}^K  - \P_n \psi_{t,\vec \eta_n}^J- P \psi_{t,\vec \eta_0}^K  + P \psi_{t,\vec \eta_0}^J  \right| \\ \leq n^{-1/2} \sup_{t \in \R} \left|  \G_n \left( \psi_{t,\vec \eta_n}^K  - \psi_{t,\vec \eta_n}^J- \psi_{t,\vec \eta_0}^K  + \psi_{t,\vec \eta_0}^J \right) \right| \\ + n^{-1/2} \sup_{t \in \R} \left|
\G_n \left( \psi_{t,\vec \eta_0}^K  - \psi_{t,\vec \eta_0}^J \right) \right| + \sup_{t \in \R} \left| P \left(\psi_{t,\vec \eta_n}^K  - \psi_{t,\vec \eta_n}^J - \psi_{t,\vec \eta_0}^K  + \psi_{t,\vec \eta_0}^J \right) \right|.
\end{multline}
Applying the mean value theorem as in the proof of Lemma~\ref{lemQuad}, we obtain that, 
$$
 \sup_{t \in \R} \left| P \left(\psi_{t,\vec \eta}^K  - \psi_{t,\vec \eta}^J - \psi_{t,\vec \eta_0}^K  + \psi_{t,\vec \eta_0}^J \right) \right| \to 0 \qquad \mbox{as} \qquad \vec \eta \to \vec \eta_0,
$$
which, combined with the fact that $\vec \eta_n \as \vec \eta_0$ implies that the last term on the right of~(\ref{tmp}) converges to zero in probability. From Lemma~\ref{lemQuad} and Theorem 2.1 in~\cite{vanWel07}, we obtain that the first term on the right of~(\ref{tmp}) converges to zero in probability. The second term on the right of~(\ref{tmp}) converges to zero in probability because the classes $\FF^J$ and $\FF^K$ are $P$-Donsker. The convergence to zero in probability of the term on the left of~(\ref{tmp}) combined with the fact that $\pi_n \as \pi_0$ and that $| \G_n' \psi_{\vec \gamma_0}^\pi |$ is bounded in probability implies that the first product in~(\ref{last}) converges to zero in probability uniformly in~$t \in \R$. Furthermore, $\FF^{\alpha,\beta,\pi}$ being $P$-Donsker, and since $P \| \Psi_{\vec \gamma_n} \Gamma_n^{-1}\dot \varphi_{\vec\gamma_n} -  \Psi_{\vec \gamma_0} \Gamma_0^{-1}\dot \varphi_{\vec\gamma_0} \|^2 \p 0$ under Assumptions A1~(ii) and A2, we have from Lemma~\ref{extended} that $\G_n' (\hat \psi_{\vec \gamma_n}^\pi - \psi_{\vec \gamma_0}^\pi) \p 0$, which implies that the second product in~(\ref{last}) converges to zero in probability uniformly in~$t \in \R$.

One can similarly show that the other terms on the right of~(\ref{diff}) converge to zero in probability uniformly in~$t \in \R$ using, among other arguments, the fact that, from Lemma~\ref{extended},
\begin{multline*}
\sup_{t \in \R} \left| \G_n' \left( \psi_{t,\vec \eta_n}^J - \psi_{t,\vec \eta_0}^J \right) \right|, \qquad \sup_{t \in \R} \left| \G_n' \left( \psi_{t,\vec \eta_n}^K - \psi_{t,\vec \eta_0}^K\right) \right|, \\ \G_n' (\hat \psi_{\vec \gamma_n}^\alpha - \psi_{\vec \gamma_0}^\alpha), \qquad \mbox{and} \qquad \G_n' (\hat \psi_{\vec \gamma_n}^\beta - \psi_{\vec \gamma_0}^\beta)
\end{multline*}
converge to zero in probability, as well as $\sup_{t \in \R} |f_n(t) - f(t)|$ since the assumptions of Proposition~\ref{pdf} are satisfied. 
\end{proof}

\bibliographystyle{plainnat}
\bibliography{biblio}

\end{document}

%% file: normnoise.tex
\begin{sidewaystable}[ht]
\begin{center}
\caption{For $M=1000$ random samples generated under scenarios WOn, MOn and SOn, number $m$ of samples out of $M$ for which $\pi_n \not \in (0,1]$, as well as estimated bias and standard deviation of $\alpha_n$, $\beta_n$, $\pi_n$, $F_n\{F^{-1}(0.1)\}$, $F_n\{F^{-1}(0.5)\}$ and $F_n\{F^{-1}(0.9)\}$ computed from the $M-m$ valid estimates.}
\label{normnoise}
\begin{tabular}{lrrrrrrrrrrrrrrr}
  \hline
   & & & & \multicolumn{2}{c}{$\alpha_n$} & \multicolumn{2}{c}{$\beta_n$} & \multicolumn{2}{c}{$\pi_n$} & \multicolumn{2}{c}{$F_n\{F^{-1}(0.1)\}$} & \multicolumn{2}{c}{$F_n\{F^{-1}(0.5)\}$} & \multicolumn{2}{c}{$F_n\{F^{-1}(0.9)\}$} \\ Scenario & $\pi_0$ & $n$ & $m$ & bias & sd & bias & sd & bias & sd & bias & sd & bias & sd & bias & sd \\ \hline
WOn & 0.4 & 100 & 15 & -0.049 & 0.689 & -0.008 & 0.340 & 0.038 & 0.139 & 0.140 & 0.144 & 0.051 & 0.160 & -0.070 & 0.119 \\ 
   &  & 300 & 0 & -0.032 & 0.392 & -0.008 & 0.220 & 0.015 & 0.079 & 0.078 & 0.092 & 0.022 & 0.129 & -0.048 & 0.098 \\ 
   &  & 1000 & 0 & -0.010 & 0.213 & -0.006 & 0.125 & 0.005 & 0.040 & 0.030 & 0.044 & 0.007 & 0.092 & -0.022 & 0.062 \\ 
   &  & 5000 & 0 & -0.005 & 0.096 & -0.002 & 0.058 & 0.002 & 0.019 & 0.008 & 0.014 & 0.000 & 0.049 & -0.007 & 0.030 \\ 
   & 0.7 & 100 & 38 & 0.015 & 0.357 & 0.019 & 0.181 & 0.003 & 0.101 & 0.060 & 0.080 & 0.035 & 0.122 & -0.024 & 0.084 \\ 
   &  & 300 & 2 & -0.011 & 0.205 & -0.002 & 0.118 & 0.010 & 0.065 & 0.025 & 0.039 & 0.009 & 0.086 & -0.018 & 0.061 \\ 
   &  & 1000 & 0 & -0.002 & 0.112 & 0.000 & 0.067 & 0.001 & 0.036 & 0.009 & 0.018 & 0.003 & 0.054 & -0.006 & 0.034 \\ 
   &  & 5000 & 0 & -0.003 & 0.050 & -0.001 & 0.030 & 0.001 & 0.017 & 0.002 & 0.006 & -0.001 & 0.027 & -0.002 & 0.015 \\ 
   \\[0.5ex]MOn & 0.4 & 100 & 34 & -0.095 & 0.827 & -0.020 & 0.376 & 0.056 & 0.153 & 0.054 & 0.088 & 0.039 & 0.099 & -0.022 & 0.068 \\ 
   &  & 300 & 0 & -0.008 & 0.456 & -0.005 & 0.237 & 0.018 & 0.089 & 0.026 & 0.054 & 0.020 & 0.068 & -0.011 & 0.049 \\ 
   &  & 1000 & 0 & -0.014 & 0.264 & -0.003 & 0.135 & 0.006 & 0.045 & 0.010 & 0.030 & 0.006 & 0.044 & -0.005 & 0.030 \\ 
   &  & 5000 & 0 & -0.004 & 0.115 & -0.004 & 0.061 & 0.002 & 0.019 & 0.002 & 0.013 & 0.001 & 0.020 & -0.002 & 0.014 \\ 
   & 0.7 & 100 & 64 & -0.008 & 0.473 & 0.020 & 0.224 & 0.008 & 0.119 & 0.018 & 0.051 & 0.023 & 0.074 & -0.005 & 0.048 \\ 
   &  & 300 & 4 & -0.014 & 0.274 & -0.005 & 0.147 & 0.012 & 0.082 & 0.011 & 0.031 & 0.006 & 0.046 & -0.005 & 0.034 \\ 
   &  & 1000 & 0 & -0.007 & 0.155 & -0.002 & 0.084 & 0.005 & 0.046 & 0.004 & 0.018 & 0.002 & 0.027 & -0.002 & 0.020 \\ 
   &  & 5000 & 0 & -0.004 & 0.069 & -0.001 & 0.038 & 0.001 & 0.021 & 0.001 & 0.007 & 0.000 & 0.012 & -0.001 & 0.009 \\ 
   \\[0.5ex]SOn & 0.4 & 100 & 251 & 0.666 & 3.963 & 0.110 & 0.393 & 0.013 & 0.222 & 0.006 & 0.153 & 0.057 & 0.122 & 0.019 & 0.053 \\ 
   &  & 300 & 90 & 0.042 & 0.522 & 0.022 & 0.230 & 0.048 & 0.183 & -0.018 & 0.047 & 0.021 & 0.051 & 0.007 & 0.028 \\ 
   &  & 1000 & 2 & -0.009 & 0.279 & 0.003 & 0.139 & 0.026 & 0.116 & -0.012 & 0.025 & 0.010 & 0.028 & 0.003 & 0.015 \\ 
   &  & 5000 & 0 & 0.005 & 0.122 & 0.002 & 0.063 & 0.003 & 0.046 & -0.002 & 0.011 & 0.002 & 0.012 & 0.001 & 0.007 \\ 
   & 0.7 & 100 & 310 & 0.199 & 0.627 & 0.112 & 0.222 & -0.057 & 0.192 & -0.016 & 0.051 & 0.021 & 0.067 & 0.014 & 0.036 \\ 
   &  & 300 & 166 & 0.090 & 0.346 & 0.040 & 0.149 & -0.019 & 0.152 & -0.011 & 0.028 & 0.008 & 0.033 & 0.006 & 0.020 \\ 
   &  & 1000 & 36 & 0.005 & 0.177 & 0.006 & 0.090 & 0.008 & 0.106 & -0.004 & 0.014 & 0.003 & 0.016 & 0.002 & 0.010 \\ 
   &  & 5000 & 0 & 0.000 & 0.084 & 0.000 & 0.043 & 0.005 & 0.053 & -0.001 & 0.006 & 0.001 & 0.007 & 0.000 & 0.005 \\ 
   \hline
\end{tabular}
\end{center}
\end{sidewaystable}

%% file: gammanoise.tex
\begin{sidewaystable}[ht]
\begin{center}
\caption{For $M=1000$ random samples generated under scenarios WOg, MOg and SOg, number $m$ of samples out of $M$ for which $\pi_n \not \in (0,1]$, as well as estimated bias and standard deviation of $\alpha_n$, $\beta_n$, $\pi_n$, $F_n\{F^{-1}(0.1)\}$, $F_n\{F^{-1}(0.5)\}$ and $F_n\{F^{-1}(0.9)\}$ computed from the $M-m$ valid estimates.}
\label{gammanoise}
\begin{tabular}{lrrrrrrrrrrrrrrr}
  \hline
   & & & & \multicolumn{2}{c}{$\alpha_n$} & \multicolumn{2}{c}{$\beta_n$} & \multicolumn{2}{c}{$\pi_n$} & \multicolumn{2}{c}{$F_n\{F^{-1}(0.1)\}$} & \multicolumn{2}{c}{$F_n\{F^{-1}(0.5)\}$} & \multicolumn{2}{c}{$F_n\{F^{-1}(0.9)\}$} \\ Scenario & $\pi_0$ & $n$ & $m$ & bias & sd & bias & sd & bias & sd & bias & sd & bias & sd & bias & sd \\ \hline
WOg & 0.4 & 100 & 21 & -0.083 & 0.651 & -0.022 & 0.342 & 0.044 & 0.134 & 0.186 & 0.167 & 0.004 & 0.164 & -0.065 & 0.108 \\ 
   &  & 300 & 0 & -0.053 & 0.381 & -0.007 & 0.225 & 0.018 & 0.082 & 0.119 & 0.127 & -0.008 & 0.134 & -0.035 & 0.083 \\ 
   &  & 1000 & 0 & -0.007 & 0.208 & -0.003 & 0.128 & 0.005 & 0.040 & 0.058 & 0.087 & -0.011 & 0.103 & -0.012 & 0.043 \\ 
   &  & 5000 & 0 & -0.004 & 0.094 & -0.002 & 0.056 & 0.002 & 0.017 & 0.016 & 0.041 & -0.006 & 0.055 & -0.003 & 0.018 \\ 
   & 0.7 & 100 & 36 & -0.014 & 0.360 & -0.009 & 0.186 & 0.018 & 0.106 & 0.098 & 0.115 & -0.008 & 0.132 & -0.024 & 0.072 \\ 
   &  & 300 & 4 & -0.009 & 0.211 & -0.005 & 0.119 & 0.008 & 0.069 & 0.056 & 0.080 & -0.010 & 0.100 & -0.013 & 0.047 \\ 
   &  & 1000 & 0 & -0.004 & 0.117 & -0.000 & 0.069 & 0.002 & 0.038 & 0.025 & 0.050 & -0.005 & 0.068 & -0.003 & 0.024 \\ 
   &  & 5000 & 0 & -0.002 & 0.051 & -0.002 & 0.031 & 0.001 & 0.017 & 0.004 & 0.023 & -0.003 & 0.031 & -0.001 & 0.010 \\ 
   \\[0.5ex]MOg & 0.4 & 100 & 45 & -0.067 & 0.846 & 0.002 & 0.400 & 0.047 & 0.156 & 0.106 & 0.122 & 0.008 & 0.112 & -0.008 & 0.056 \\ 
   &  & 300 & 0 & -0.049 & 0.458 & -0.015 & 0.249 & 0.024 & 0.095 & 0.061 & 0.079 & -0.001 & 0.079 & -0.006 & 0.035 \\ 
   &  & 1000 & 0 & -0.025 & 0.248 & -0.012 & 0.141 & 0.008 & 0.045 & 0.024 & 0.044 & -0.008 & 0.052 & -0.003 & 0.020 \\ 
   &  & 5000 & 0 & -0.006 & 0.115 & -0.002 & 0.064 & 0.002 & 0.020 & 0.006 & 0.019 & -0.002 & 0.026 & -0.000 & 0.009 \\ 
   & 0.7 & 100 & 69 & -0.011 & 0.511 & 0.007 & 0.222 & 0.018 & 0.124 & 0.049 & 0.081 & -0.001 & 0.084 & 0.000 & 0.037 \\ 
   &  & 300 & 7 & -0.031 & 0.299 & -0.004 & 0.153 & 0.016 & 0.089 & 0.029 & 0.049 & -0.005 & 0.059 & -0.002 & 0.023 \\ 
   &  & 1000 & 0 & -0.008 & 0.163 & -0.003 & 0.087 & 0.006 & 0.049 & 0.011 & 0.027 & -0.003 & 0.036 & -0.001 & 0.012 \\ 
   &  & 5000 & 0 & 0.002 & 0.071 & 0.001 & 0.040 & 0.000 & 0.022 & 0.003 & 0.011 & -0.000 & 0.017 & 0.000 & 0.006 \\ 
   \\[0.5ex]SOg & 0.4 & 100 & 305 & 1.339 & 12.672 & 0.155 & 0.455 & 0.012 & 0.224 & 0.062 & 0.190 & 0.024 & 0.138 & 0.021 & 0.049 \\ 
   &  & 300 & 145 & 0.076 & 0.619 & 0.055 & 0.274 & 0.041 & 0.182 & 0.018 & 0.087 & 0.001 & 0.060 & 0.010 & 0.024 \\ 
   &  & 1000 & 21 & -0.011 & 0.314 & -0.000 & 0.168 & 0.035 & 0.132 & 0.005 & 0.042 & -0.000 & 0.032 & 0.003 & 0.013 \\ 
   &  & 5000 & 0 & -0.004 & 0.152 & -0.000 & 0.079 & 0.011 & 0.062 & 0.002 & 0.018 & -0.000 & 0.014 & 0.001 & 0.006 \\ 
   & 0.7 & 100 & 386 & 1.222 & 22.682 & 0.169 & 0.326 & -0.085 & 0.207 & 0.043 & 0.117 & 0.020 & 0.079 & 0.009 & 0.036 \\ 
   &  & 300 & 244 & 0.101 & 0.379 & 0.069 & 0.189 & -0.028 & 0.167 & 0.017 & 0.051 & 0.005 & 0.037 & 0.003 & 0.017 \\ 
   &  & 1000 & 75 & 0.021 & 0.206 & 0.018 & 0.117 & 0.003 & 0.126 & 0.005 & 0.028 & 0.001 & 0.021 & 0.002 & 0.010 \\ 
   &  & 5000 & 0 & -0.003 & 0.100 & -0.000 & 0.055 & 0.007 & 0.067 & 0.001 & 0.012 & 0.000 & 0.009 & 0.000 & 0.004 \\ 
   \hline
\end{tabular}
\end{center}
\end{sidewaystable}

%% file: expnoise.tex
\begin{sidewaystable}[ht]
\begin{center}
\caption{For $M=1000$ random samples generated under scenarios WOe, MOe and SOe, number $m$ of samples out of $M$ for which $\pi_n \not \in (0,1]$, as well as estimated bias and standard deviation of $\alpha_n$, $\beta_n$, $\pi_n$, $F_n\{F^{-1}(0.1)\}$, $F_n\{F^{-1}(0.5)\}$ and $F_n\{F^{-1}(0.9)\}$ computed from the $M-m$ valid estimates.}
\label{expnoise}
\begin{tabular}{lrrrrrrrrrrrrrrr}
  \hline
   & & & & \multicolumn{2}{c}{$\alpha_n$} & \multicolumn{2}{c}{$\beta_n$} & \multicolumn{2}{c}{$\pi_n$} & \multicolumn{2}{c}{$F_n\{F^{-1}(0.1)\}$} & \multicolumn{2}{c}{$F_n\{F^{-1}(0.5)\}$} & \multicolumn{2}{c}{$F_n\{F^{-1}(0.9)\}$} \\ Scenario & $\pi_0$ & $n$ & $m$ & bias & sd & bias & sd & bias & sd & bias & sd & bias & sd & bias & sd \\ \hline
WOe & 0.4 & 100 & 26 & -0.040 & 0.715 & -0.027 & 0.336 & 0.045 & 0.138 & 0.224 & 0.185 & -0.008 & 0.179 & -0.060 & 0.106 \\ 
   &  & 300 & 0 & -0.017 & 0.380 & -0.005 & 0.218 & 0.013 & 0.074 & 0.154 & 0.152 & -0.021 & 0.151 & -0.031 & 0.077 \\ 
   &  & 1000 & 0 & -0.009 & 0.215 & -0.003 & 0.125 & 0.004 & 0.040 & 0.084 & 0.115 & -0.025 & 0.118 & -0.011 & 0.041 \\ 
   &  & 5000 & 0 & -0.003 & 0.092 & 0.001 & 0.055 & 0.001 & 0.017 & 0.028 & 0.073 & -0.010 & 0.066 & -0.002 & 0.015 \\ 
   & 0.7 & 100 & 47 & 0.000 & 0.372 & 0.007 & 0.189 & 0.013 & 0.108 & 0.145 & 0.149 & -0.017 & 0.149 & -0.021 & 0.071 \\ 
   &  & 300 & 1 & -0.017 & 0.203 & -0.001 & 0.126 & 0.010 & 0.071 & 0.085 & 0.113 & -0.021 & 0.116 & -0.011 & 0.046 \\ 
   &  & 1000 & 0 & -0.006 & 0.111 & -0.004 & 0.070 & 0.003 & 0.037 & 0.036 & 0.079 & -0.017 & 0.079 & -0.004 & 0.022 \\ 
   &  & 5000 & 0 & -0.002 & 0.051 & 0.000 & 0.031 & 0.001 & 0.017 & 0.009 & 0.049 & -0.004 & 0.039 & -0.000 & 0.009 \\ 
   \\[0.5ex]MOe & 0.4 & 100 & 44 & -0.020 & 1.104 & -0.005 & 0.390 & 0.047 & 0.153 & 0.148 & 0.146 & -0.008 & 0.128 & -0.011 & 0.052 \\ 
   &  & 300 & 0 & -0.040 & 0.463 & -0.005 & 0.259 & 0.019 & 0.090 & 0.092 & 0.109 & -0.017 & 0.097 & -0.005 & 0.034 \\ 
   &  & 1000 & 0 & -0.012 & 0.255 & -0.005 & 0.146 & 0.007 & 0.046 & 0.043 & 0.073 & -0.013 & 0.067 & -0.001 & 0.019 \\ 
   &  & 5000 & 0 & -0.005 & 0.115 & -0.003 & 0.065 & 0.002 & 0.021 & 0.010 & 0.042 & -0.004 & 0.034 & -0.001 & 0.008 \\ 
   & 0.7 & 100 & 82 & -0.021 & 0.498 & 0.014 & 0.242 & 0.015 & 0.127 & 0.081 & 0.120 & -0.018 & 0.100 & -0.000 & 0.036 \\ 
   &  & 300 & 4 & -0.012 & 0.289 & -0.002 & 0.155 & 0.012 & 0.086 & 0.048 & 0.082 & -0.013 & 0.073 & -0.001 & 0.022 \\ 
   &  & 1000 & 0 & -0.002 & 0.162 & -0.001 & 0.090 & 0.004 & 0.050 & 0.022 & 0.057 & -0.006 & 0.048 & -0.001 & 0.012 \\ 
   &  & 5000 & 0 & -0.002 & 0.069 & -0.002 & 0.040 & 0.001 & 0.022 & 0.002 & 0.030 & -0.002 & 0.021 & -0.000 & 0.006 \\ 
   \\[0.5ex]SOe & 0.4 & 100 & 325 & 0.972 & 7.133 & 0.191 & 0.533 & 0.008 & 0.220 & 0.104 & 0.205 & -0.000 & 0.146 & 0.015 & 0.053 \\ 
   &  & 300 & 194 & 0.049 & 0.600 & 0.044 & 0.276 & 0.051 & 0.192 & 0.047 & 0.109 & -0.013 & 0.074 & 0.007 & 0.027 \\ 
   &  & 1000 & 36 & -0.014 & 0.342 & 0.005 & 0.177 & 0.045 & 0.147 & 0.029 & 0.074 & -0.011 & 0.050 & 0.004 & 0.015 \\ 
   &  & 5000 & 0 & -0.001 & 0.160 & 0.002 & 0.087 & 0.009 & 0.066 & 0.010 & 0.042 & -0.002 & 0.025 & 0.001 & 0.007 \\ 
   & 0.7 & 100 & 399 & 0.432 & 1.880 & 0.213 & 0.437 & -0.097 & 0.211 & 0.090 & 0.155 & 0.016 & 0.096 & 0.006 & 0.036 \\ 
   &  & 300 & 299 & 0.133 & 0.398 & 0.091 & 0.213 & -0.043 & 0.170 & 0.048 & 0.094 & 0.007 & 0.054 & 0.001 & 0.018 \\ 
   &  & 1000 & 97 & 0.031 & 0.230 & 0.019 & 0.121 & 0.004 & 0.135 & 0.021 & 0.061 & 0.000 & 0.034 & 0.001 & 0.010 \\ 
   &  & 5000 & 1 & -0.004 & 0.110 & -0.001 & 0.061 & 0.011 & 0.077 & 0.004 & 0.031 & -0.001 & 0.016 & 0.001 & 0.005 \\ 
   \hline
\end{tabular}
\end{center}
\end{sidewaystable}

%% file: stdall.tex
\begin{sidewaystable}[ht]
\begin{center}
\caption{For $M=1000$ random samples generated under scenarios WOn, MOg and SOe, number $m$ of samples out of $M$ for which $\pi_n \not \in (0,1]$, and, for each of the estimators $\alpha_n$, $\beta_n$, $\pi_n$, $F_n\{F^{-1}(0.1)\}$, $F_n\{F^{-1}(0.5)\}$ and $F_n\{F^{-1}(0.9)\}$, standard deviation of the $M-m$ valid estimates times $\sqrt{n}$, and mean of the estimated standard errors times $\sqrt{n}$. The quantities $t_1$, $t_2$ and $t_3$ in the table are equal to $F^{-1}(0.1)$, $F^{-1}(0.5)$ and $F^{-1}(0.9)$, respectively. }
\label{stdall}
\begin{tabular}{lrrrrrrrrrrrrrrr}
  \hline
   & & & & \multicolumn{2}{c}{$\alpha_n$} & \multicolumn{2}{c}{$\beta_n$} & \multicolumn{2}{c}{$\pi_n$} & \multicolumn{2}{c}{$F_n(t_1)$} & \multicolumn{2}{c}{$F_n(t_2)$} & \multicolumn{2}{c}{$F_n(t_3)$} \\ Scenario & $\pi_0$ & $n$ & $m$ & sd & $\overline{\mbox{se}}$ & sd & $\overline{\mbox{se}}$ & sd & $\overline{\mbox{se}}$ & sd & $\overline{\mbox{se}}$ & sd & $\overline{\mbox{se}}$ & sd & $\overline{\mbox{se}}$  \\ \hline
WOn & 0.4 & 100 & 16 & 6.66 & 6.67 & 3.51 & 2.92 & 1.37 & 1.23 & 1.43 & 1.15 & 1.57 & 1.36 & 1.18 & 1.11 \\ 
   &  & 300 & 0 & 7.10 & 6.49 & 3.88 & 3.43 & 1.42 & 1.23 & 1.55 & 1.18 & 2.22 & 1.90 & 1.72 & 1.50 \\ 
   &  & 1000 & 0 & 6.63 & 6.56 & 4.09 & 3.79 & 1.30 & 1.22 & 1.46 & 1.09 & 2.88 & 2.62 & 1.97 & 1.81 \\ 
   &  & 5000 & 0 & 6.42 & 6.61 & 4.00 & 3.92 & 1.19 & 1.24 & 0.95 & 0.86 & 3.31 & 3.23 & 1.88 & 1.93 \\ 
   &  & 25000 & 0 & 6.74 & 6.62 & 3.98 & 3.96 & 1.25 & 1.24 & 0.78 & 0.75 & 3.55 & 3.44 & 1.94 & 1.92 \\ 
   & 0.7 & 100 & 33 & 3.49 & 3.50 & 1.86 & 1.61 & 1.04 & 1.05 & 0.73 & 0.60 & 1.16 & 1.02 & 0.87 & 0.75 \\ 
   &  & 300 & 2 & 3.56 & 3.54 & 2.08 & 1.89 & 1.19 & 1.12 & 0.71 & 0.56 & 1.49 & 1.34 & 1.07 & 0.93 \\ 
   &  & 1000 & 0 & 3.77 & 3.58 & 2.17 & 2.08 & 1.23 & 1.17 & 0.56 & 0.50 & 1.82 & 1.65 & 1.17 & 1.05 \\ 
   &  & 5000 & 0 & 3.60 & 3.63 & 2.16 & 2.18 & 1.18 & 1.20 & 0.45 & 0.43 & 1.89 & 1.88 & 1.08 & 1.09 \\ 
   &  & 25000 & 0 & 3.60 & 3.61 & 2.12 & 2.17 & 1.18 & 1.19 & 0.41 & 0.41 & 1.94 & 1.92 & 1.04 & 1.07 \\ 
   \\[0.3ex]MOg & 0.4 & 100 & 57 & 7.96 & 7.91 & 3.92 & 3.33 & 1.53 & 1.46 & 1.15 & 1.03 & 1.11 & 1.08 & 0.54 & 0.62 \\ 
   &  & 300 & 2 & 7.99 & 7.69 & 4.41 & 3.93 & 1.60 & 1.39 & 1.43 & 1.09 & 1.38 & 1.32 & 0.61 & 0.64 \\ 
   &  & 1000 & 0 & 8.37 & 7.83 & 4.64 & 4.34 & 1.50 & 1.40 & 1.46 & 1.10 & 1.74 & 1.63 & 0.64 & 0.65 \\ 
   &  & 5000 & 0 & 8.39 & 8.04 & 4.69 & 4.54 & 1.52 & 1.43 & 1.38 & 1.13 & 1.96 & 1.86 & 0.65 & 0.64 \\ 
   &  & 25000 & 0 & 8.30 & 8.04 & 4.57 & 4.58 & 1.52 & 1.44 & 1.28 & 1.19 & 1.96 & 1.91 & 0.65 & 0.64 \\ 
   & 0.7 & 100 & 66 & 4.55 & 4.70 & 2.47 & 2.07 & 1.27 & 1.26 & 0.86 & 0.65 & 0.82 & 0.77 & 0.37 & 0.39 \\ 
   &  & 300 & 8 & 5.06 & 4.80 & 2.71 & 2.42 & 1.51 & 1.40 & 0.89 & 0.70 & 1.03 & 0.95 & 0.41 & 0.41 \\ 
   &  & 1000 & 0 & 5.05 & 4.95 & 2.73 & 2.64 & 1.57 & 1.48 & 0.86 & 0.70 & 1.15 & 1.10 & 0.43 & 0.42 \\ 
   &  & 5000 & 0 & 5.00 & 5.01 & 2.72 & 2.73 & 1.55 & 1.52 & 0.79 & 0.73 & 1.17 & 1.17 & 0.41 & 0.42 \\ 
   &  & 25000 & 0 & 4.93 & 5.03 & 2.71 & 2.76 & 1.52 & 1.53 & 0.79 & 0.78 & 1.19 & 1.19 & 0.42 & 0.42 \\ 
   \\[0.3ex]SOe & 0.4 & 100 & 294 & 76.74 & 60.97 & 6.19 & 4.65 & 2.24 & 3.59 & 1.94 & 2.30 & 1.36 & 1.94 & 0.51 & 0.80 \\ 
   &  & 300 & 171 & 11.91 & 10.92 & 5.13 & 4.92 & 3.40 & 4.35 & 2.13 & 1.64 & 1.40 & 1.55 & 0.46 & 0.60 \\ 
   &  & 1000 & 31 & 11.20 & 10.24 & 6.05 & 5.52 & 4.65 & 4.65 & 2.47 & 1.79 & 1.62 & 1.58 & 0.49 & 0.53 \\ 
   &  & 5000 & 0 & 11.47 & 10.87 & 6.17 & 5.93 & 4.64 & 4.38 & 2.91 & 2.47 & 1.70 & 1.68 & 0.48 & 0.48 \\ 
   &  & 25000 & 0 & 10.96 & 11.23 & 6.06 & 6.16 & 4.27 & 4.37 & 3.68 & 3.49 & 1.64 & 1.72 & 0.46 & 0.47 \\ 
   & 0.7 & 100 & 410 & 8.91 & 8.82 & 3.37 & 3.43 & 2.06 & 3.00 & 1.48 & 1.19 & 0.87 & 1.11 & 0.36 & 0.44 \\ 
   &  & 300 & 262 & 7.58 & 7.51 & 4.07 & 4.00 & 3.06 & 4.02 & 1.75 & 1.36 & 0.96 & 1.13 & 0.33 & 0.39 \\ 
   &  & 1000 & 121 & 7.41 & 7.55 & 4.09 & 4.23 & 4.44 & 5.04 & 1.92 & 1.54 & 1.07 & 1.19 & 0.31 & 0.36 \\ 
   &  & 5000 & 1 & 8.06 & 7.83 & 4.38 & 4.35 & 5.58 & 5.43 & 2.33 & 2.11 & 1.20 & 1.19 & 0.34 & 0.34 \\ 
   &  & 25000 & 0 & 8.00 & 8.00 & 4.36 & 4.45 & 5.44 & 5.50 & 2.80 & 2.76 & 1.22 & 1.22 & 0.33 & 0.34 \\ 
   \hline
\end{tabular}
\end{center}
\end{sidewaystable}

%% file: cball.tex
\begin{table}[ht]
\begin{center}
\caption{For $M=1000$ random samples generated under each of the nine scenarios considered in Section~\ref{mc}, number $m$ of samples out of $M$ for which $\pi_n \not \in (0,1]$, and proportion $p$ out of the $M-m$ remaining samples for which $F_n$ is not in the approximate confidence band computed as explained in Subsection~\ref{bootstrap}.}
\label{cball}
\begin{tabular}{lrrrrrrrr}
  \hline
  Generic & & & \multicolumn{2}{c}{$\varepsilon\sim$~Normal} & \multicolumn{2}{c}{$\varepsilon\sim$~Gamma} & \multicolumn{2}{c}{$\varepsilon\sim$~Exp} \\ scenario & $\pi_0$ & $n$ & $m$ & $p$ & $m$ & $p$ & $m$ & $p$ \\ \hline
WO & 0.4 & 100 & 22 & 0.306 &   27 & 0.362 &   24 & 0.444 \\ 
   &  & 300 & 0 & 0.238 &    0 & 0.251 &    2 & 0.334 \\ 
   &  & 1000 & 0 & 0.126 &    0 & 0.182 &    0 & 0.226 \\ 
   &  & 5000 & 0 & 0.082 &    0 & 0.080 &    0 & 0.133 \\ 
   &  & 25000 & 0 & 0.064 &    0 & 0.055 &    0 & 0.092 \\ 
   & 0.7 & 100 & 32 & 0.169 &   32 & 0.195 &   24 & 0.290 \\ 
   &  & 300 & 2 & 0.138 &    5 & 0.160 &    3 & 0.231 \\ 
   &  & 1000 & 0 & 0.092 &    0 & 0.108 &    0 & 0.168 \\ 
   &  & 5000 & 0 & 0.073 &    0 & 0.074 &    0 & 0.090 \\ 
   &  & 25000 & 0 & 0.056 &    0 & 0.041 &    0 & 0.081 \\ 
   \\[0.5ex]MO & 0.4 & 100 & 45 & 0.088 &   42 & 0.177 &   48 & 0.334 \\ 
   &  & 300 & 0 & 0.114 &    2 & 0.205 &    1 & 0.296 \\ 
   &  & 1000 & 0 & 0.103 &    0 & 0.127 &    0 & 0.207 \\ 
   &  & 5000 & 0 & 0.073 &    0 & 0.095 &    0 & 0.126 \\ 
   &  & 25000 & 0 & 0.050 &    0 & 0.073 &    0 & 0.085 \\ 
   & 0.7 & 100 & 76 & 0.088 &   60 & 0.117 &   67 & 0.247 \\ 
   &  & 300 & 7 & 0.102 &   13 & 0.146 &   12 & 0.215 \\ 
   &  & 1000 & 0 & 0.084 &    0 & 0.082 &    0 & 0.140 \\ 
   &  & 5000 & 0 & 0.054 &    0 & 0.067 &    0 & 0.096 \\ 
   &  & 25000 & 0 & 0.049 &    0 & 0.065 &    0 & 0.070 \\ 
   \\[0.5ex]SO & 0.4 & 100 & 259 & 0.003 &  327 & 0.030 &  316 & 0.072 \\ 
   &  & 300 & 103 & 0.006 &  128 & 0.057 &  182 & 0.117 \\ 
   &  & 1000 & 4 & 0.027 &   14 & 0.067 &   29 & 0.142 \\ 
   &  & 5000 & 0 & 0.029 &    0 & 0.077 &    0 & 0.123 \\ 
   &  & 25000 & 0 & 0.042 &    0 & 0.045 &    0 & 0.087 \\ 
   & 0.7 & 100 & 328 & 0.001 &  413 & 0.036 &  405 & 0.099 \\ 
   &  & 300 & 166 & 0.005 &  249 & 0.037 &  280 & 0.094 \\ 
   &  & 1000 & 32 & 0.028 &   91 & 0.043 &  119 & 0.083 \\ 
   &  & 5000 & 0 & 0.036 &    2 & 0.062 &    2 & 0.088 \\ 
   &  & 25000 & 0 & 0.044 &    0 & 0.061 &    0 & 0.071 \\ 
   \hline
\end{tabular}
\end{center}
\end{table}

%% file: mr.bbl
\begin{thebibliography}{31}
\providecommand{\natexlab}[1]{#1}
\providecommand{\url}[1]{\texttt{#1}}
\expandafter\ifx\csname urlstyle\endcsname\relax
  \providecommand{\doi}[1]{doi: #1}\else
  \providecommand{\doi}{doi: \begingroup \urlstyle{rm}\Url}\fi

\bibitem[Anderson(1979)]{And79}
J.A. Anderson.
\newblock Multivariate logistic compounds.
\newblock \emph{Biometrika}, pages 17--26, 1979.

\bibitem[Boiteau et~al.(1998)Boiteau, Singh, Singh, Tai, and
  Turner]{BoiSinSinTaiTur98}
G.~Boiteau, M.~Singh, R.P. Singh, G.C.C. Tai, and T.R. Turner.
\newblock Rate of spread of {PVY}-n by alate {M}yzus persicae ({S}ulzer) from
  infected to healthy plants under laboratory conditions.
\newblock \emph{Potato Research}, 41:\penalty0 335--344, 1998.

\bibitem[Bordes et~al.(2006)Bordes, Delmas, and Vandekerkhove]{BorDelVan06}
L.~Bordes, C.~Delmas, and P.~Vandekerkhove.
\newblock Estimating a two-component mixture model when a component is known.
\newblock \emph{Scandinavian Journal of Statistics}, 33\penalty0 (4):\penalty0
  733--752, 2006.

\bibitem[Cohen(1980)]{Coh80}
E.A. Cohen.
\newblock \emph{Inharmonic Tone Perception}.
\newblock PhD thesis, Stanford University, 1980.

\bibitem[{De Veaux}(1989)]{DeV89}
R.D. {De Veaux}.
\newblock Mixtures of linear regressions.
\newblock \emph{Computational Statistics and Data Analysis}, 8:\penalty0
  227--245, 1989.

\bibitem[Duong(2012)]{ks}
T.~Duong.
\newblock \emph{ks: Kernel smoothing}, 2012.
\newblock URL \url{http://CRAN.R-project.org/package=ks}.
\newblock R package version 1.8.8.

\bibitem[Glad et~al.(2003)Glad, Hjort, and Ushakov]{GlaHjoUsh03}
I.K. Glad, N.L. Hjort, and N.G. Ushakov.
\newblock Correction of density estimators that are not densities.
\newblock \emph{Scandinavian Journal of Statistics}, 30:\penalty0 415--427,
  2003.

\bibitem[Gr\"un and Leisch(2006)]{GruLei06}
B.~Gr\"un and F.~Leisch.
\newblock Fitting finite mixtures of linear regression models with varying and
  fixed effects in \textsf{R}.
\newblock In A.~Rizzi and M.~Vichi, editors, \emph{Compstat 2006, Proceedings
  in Computational Statistics}, pages 853--860. Physica Verlag, Heidelberg,
  Germany, 2006.

\bibitem[Hall and Zhou(2003)]{HalZho03}
P.~Hall and X-H. Zhou.
\newblock Nonparametric estimation of component distributions in a multivariate
  mixture.
\newblock \emph{Annals of Statistics}, 31:\penalty0 201--224, 2003.

\bibitem[Hawkins et~al.(2001)Hawkins, Allen, and Stomberg]{HawAllSto01}
D.S. Hawkins, D.M. Allen, and A.J. Stomberg.
\newblock Determining the number of components in mixtures of linear models.
\newblock \emph{Computational Statistics and Data Analysis}, 38:\penalty0
  15--48, 2001.

\bibitem[Hunter and Young(2012)]{HunYou12}
D.R. Hunter and D.S. Young.
\newblock Semiparametric mixtures of regressions.
\newblock \emph{Journal of Nonparametric Statistics}, pages 19--38, 2012.

\bibitem[Hurn et~al.(2003)Hurn, Justel, and Robert]{HurJusRob03}
M.~Hurn, A.~Justel, and C.P. Robert.
\newblock Estimating mixtures of regressions.
\newblock \emph{Journal of Computational and Graphical Statistiscs},
  12:\penalty0 1--25, 2003.

\bibitem[Jones and McLachlan(1992)]{JonMcL92}
P.N. Jones and G.J. McLachlan.
\newblock Fitting finite mixture models in a regression context.
\newblock \emph{Australian Journal of Statistics}, 34:\penalty0 233--240, 1992.

\bibitem[Kosorok(2008)]{Kos08}
M.R. Kosorok.
\newblock \emph{Introduction to empirical processes and semiparametric
  inference}.
\newblock Springer, New York, 2008.

\bibitem[Leisch(2004)]{Lei04}
F.~Leisch.
\newblock Flexmix: A general framework for finite mixture models and latent
  class regression in \textsf{R}.
\newblock \emph{Journal of Statistical Software}, 2004.
\newblock http://www.jstatsoft.org/v11/i08/.

\bibitem[Leung and Qin(2006)]{LeuQin06}
D.H-Y. Leung and J.~Qin.
\newblock Semi-parametric inference in a bivariate (multivariate) mixture
  model.
\newblock \emph{Statistica Sinica}, 16:\penalty0 153--163, 2006.

\bibitem[Martin-Magniette et~al.(2008)Martin-Magniette, Mary-Huard, B\'erard,
  and Robin]{MarMarBer08}
M-L. Martin-Magniette, T.~Mary-Huard, C.~B\'erard, and S.~Robin.
\newblock Ch{IP}mix: {M}ixture model of regressions for two-color {C}h{IP}-chip
  analysis.
\newblock \emph{Bioinformatics}, 24:\penalty0 181--186, 2008.

\bibitem[Nolan and Pollard(1987)]{NolPol87}
D.~Nolan and D.~Pollard.
\newblock {$U$}-processes: {R}ates of convergence.
\newblock \emph{Annals of Statistics}, 15:\penalty0 780--799, 1987.

\bibitem[Quandt and Ramsey(1978)]{QuaRam78}
R.~Quandt and J.~Ramsey.
\newblock Estimating mixtures of normal distributions and switching regression.
\newblock \emph{Journal of the American Statistical Association}, 73:\penalty0
  730--738, 1978.

\bibitem[{R Development Core Team}(2013)]{Rsystem}
{R Development Core Team}.
\newblock \emph{{R}: {A} Language and Environment for Statistical Computing}.
\newblock R Foundation for Statistical Computing, Vienna, Austria, 2013.
\newblock URL \url{http://www.R-project.org}.
\newblock {ISBN} 3-900051-07-0.

\bibitem[St\"adler et~al.(2010)St\"adler, B\"uhlmann, and {van de
  Geer}]{StaBuhvan10}
N.~St\"adler, P.~B\"uhlmann, and S.~{van de Geer}.
\newblock $\ell_1$-penalization for mixture of regression models.
\newblock \emph{Test}, 19:\penalty0 209--256, 2010.

\bibitem[Turner(2011)]{mixreg}
R.~Turner.
\newblock \emph{mixreg: {F}unctions to fit mixtures of regressions}, 2011.
\newblock URL \url{http://CRAN.R-project.org/package=mixreg}.
\newblock R package version 0.0-4.

\bibitem[Turner(2000)]{Tur00}
T.R. Turner.
\newblock Estimating the propagation rate of a viral infection of potato plants
  via mixtures of regressions.
\newblock \emph{Applied Statistics}, 49:\penalty0 371--384, 2000.

\bibitem[{van der Vaart}(1998)]{van98}
A.W. {van der Vaart}.
\newblock \emph{Asymptotic statistics}.
\newblock Cambridge University Press, 1998.

\bibitem[{van der Vaart}(2002)]{van02}
A.W. {van der Vaart}.
\newblock Semiparametric statistics.
\newblock In \emph{\'Ecole d'\'et\'e de Saint-Flour 1999}, pages 331--457.
  Springer, New-York, 2002.

\bibitem[{van der Vaart} and Wellner(2000)]{vanWel96}
A.W. {van der Vaart} and J.A. Wellner.
\newblock \emph{Weak convergence and empirical processes}.
\newblock Springer, New York, 2000.
\newblock Second edition.

\bibitem[{van der Vaart} and Wellner(2007)]{vanWel07}
A.W. {van der Vaart} and J.A. Wellner.
\newblock Empirical processes indexed by estimated functions.
\newblock In \emph{Asymptotics: {P}articles, {P}rocesses and {I}nverse
  Problems}, pages 234--252. Institute of Mathematical Statistics, 2007.

\bibitem[Vandekerkhove(2013)]{Van13}
P.~Vandekerkhove.
\newblock Estimation of a semiparametric mixture of regressions model.
\newblock \emph{Journal of Nonparametric Statistics}, 25\penalty0 (1):\penalty0
  181--208, 2013.

\bibitem[Wand and Jones(1994)]{WanJon94}
M.P. Wand and M.C. Jones.
\newblock Multivariate plugin bandwidth selection.
\newblock \emph{Computational Statistics}, 9:\penalty0 97--116, 1994.

\bibitem[Young and Hunter(2010)]{YouHun10}
D.S. Young and D.R. Hunter.
\newblock Mixtures of regressions with predictor-dependent mixing proportions.
\newblock \emph{Computational Statistics and Data Analysis}, pages 2253--2266,
  2010.

\bibitem[Zhu and Zhang(2004)]{ZhuZha04}
H.~Zhu and H.~Zhang.
\newblock Hypothesis testing in mixture regression models.
\newblock \emph{Journal of the Royal Statistical Society Series B},
  66:\penalty0 3--16, 2004.

\end{thebibliography}
